\documentclass[11pt,letterpaper]{amsart}
\usepackage[foot]{amsaddr} 
\usepackage[in]{fullpage}
\usepackage{amssymb}

\usepackage{microtype}
\usepackage[T1]{fontenc}
\usepackage[mono=false]{libertine}
\usepackage{libertinust1math}
\usepackage{hyperref}
\hypersetup{colorlinks=true, citecolor=blue, linkcolor = red, urlcolor=blue}
\usepackage{enumerate}
\usepackage{graphicx}
\graphicspath{{figures/}}
\usepackage{subcaption}
\usepackage{multirow}
\usepackage{array}
\newcommand{\ra}[1]{\renewcommand{\arraystretch}{#1}}
\usepackage{booktabs}
\usepackage{pbox} % change line in a table cell
\usepackage{float} % for fix the order of figure placement

\newtheorem{theorem}{Theorem}[section]
\newtheorem*{theorem*}{Theorem}
\newtheorem{corollary}[theorem]{Corollary}
\newtheorem{lemma}[theorem]{Lemma}
\newtheorem*{lemma*}{Lemma}
\newtheorem{proposition}[theorem]{Proposition}

\theoremstyle{definition}

\newtheorem*{definition*}{Definition}
\theoremstyle{remark}
\newtheorem{remark}{Remark}[section]
\newtheorem*{notation}{Notation}

\numberwithin{equation}{section}

\newcommand{\figref}[1]{Figure~\ref{#1}}
\newcommand{\secref}[1]{Section~\ref{#1}}
\newcommand{\thmref}[1]{Theorem~\ref{#1}}
\newcommand{\lemref}[1]{Lemma~\ref{#1}}
\newcommand{\propref}[1]{Proposition~\ref{#1}}

\newcommand{\corref}[1]{Corollary~\ref{#1}}
\newcommand{\tabref}[1]{Table~\ref{#1}}
\newcommand{\appref}[1]{Appendix~\ref{#1}}

\begin{document}

% for \documentclass[11pt,a4paper]{amsart}
%\title{FPRAS via MCMC when rapid mixing proof fails}
%\title{FPRAS via MCMC where it mixes torpidly (and very little effort)}
\title{FPRAS VIA MCMC WHERE IT MIXES TORPIDLY (and very little effort)}

\author{Jin-Yi Cai}
\email{jyc@cs.wisc.edu}
\author{Tianyu Liu}
\email{tl@cs.wisc.edu}
\address{Department of Computer Sciences, University of Wisconsin--Madison.}
\thanks{Supported by NSF CCF-1714275}

\begin{abstract}
Is Fully Polynomial-time Randomized Approximation Scheme (FPRAS)
for a problem via an
MCMC algorithm possible when it is known that 
rapid
mixing provably fails? We introduce several weight-preserving maps for the eight-vertex model on planar and on
bipartite graphs, respectively. Some are one-to-one, while others are holographic which map 
superpositions of exponentially many states from one setting to another,
in a quantum-like many-to-many fashion.
%In fact we introduce a group of such mappings in each case.
In fact we introduce a set of such mappings that forms a group in each case.
Using some holographic maps and their compositions we obtain
FPRAS for the eight-vertex model at parameter settings where it is known that 
rapid mixing provably fails
due to an intrinsic barrier. This FPRAS is indeed the same MCMC algorithm, except its state space corresponds
to superpositions of the given states, where rapid mixing holds.
%We establish FPRAS for the eight-vertex model on planar 4-regular graphs under some parameter settings and on bipartite graphs under some other parameter settings where the same problems on general 4-regular graphs are NP-hard to approximate. 
%FPRAS is also given for torus graphs under these parameter settings where natural Markov chains are known to mix torpidly.
%Our results 
%provide the first problems that provably possess the following properties:
%they are NP-hard to approximate in general but FPRASable on planar and
%bipartite graphs respectively.
FPRAS is also given for torus graphs for parameter settings where natural Markov chains are known to mix torpidly. 
Our results  show that  the eight-vertex model is
the first problem with the provable property  that
while NP-hard to approximate on general graphs (even \#P-hard for
planar graphs in exact complexity),
it possesses FPRAS on both bipartite graphs and
planar graphs in substantial regions of its parameter space.
\end{abstract}

% for amsart
%\maketitle

\begingroup
\def\uppercasenonmath#1{} % this disables uppercasing title
\maketitle
\endgroup

\thispagestyle{empty}
\clearpage
\setcounter{page}{1}

\section{Introduction}\label{sec:intro}
Let $G$ be any 4-regular graph. We label
four incident edges of each vertex from 1 to 4.
The eight-vertex model on $G$ is defined as follows.
The states consist of \emph{even orientations}, i.e. all orientations
having an even number of arrows into (and out of) each vertex.
There are eight permitted types of local configurations around a vertex---hence the name eight-vertex model (see \figref{fig:orientations}).%In the \emph{unweighted} case, the problem is to
%count the number of even orientations of $G$, and this
%is computable in polynomial time.
%In the general case
%of the eight-vertex model there are 
%\emph{weights} associated with 
%local configurations,
%and the problem is to compute a weighted sum called
%the partition function. This becomes an  interesting
%and challenging problem, and the complexity
%picture becomes  more intricate~\cite{DBLP:journals/corr/CaiF17}.

\captionsetup[subfigure]{labelformat=empty}
\renewcommand{\thesubfigure}{-\arabic{subfigure}}
\begin{figure}[h!]
\centering
\begin{subfigure}[b]{0.12\linewidth}
\centering\includegraphics[width=\linewidth]{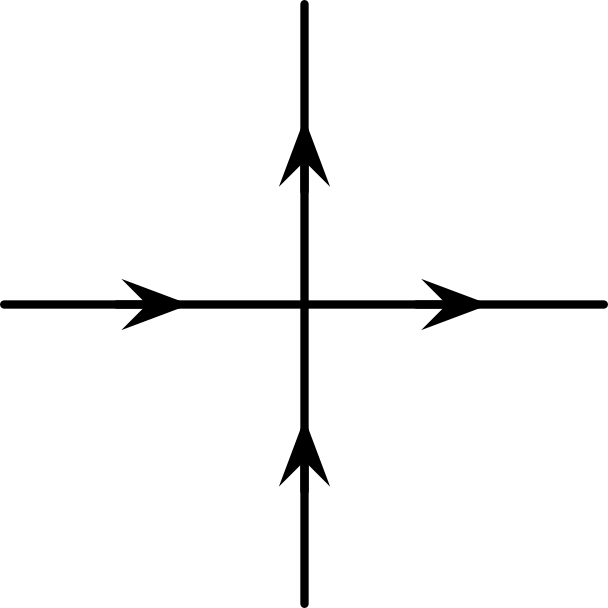}\caption{$1$}
\label{fig:orientations_1}
\end{subfigure}
\begin{subfigure}[b]{0.12\linewidth}
\centering\includegraphics[width=\linewidth]{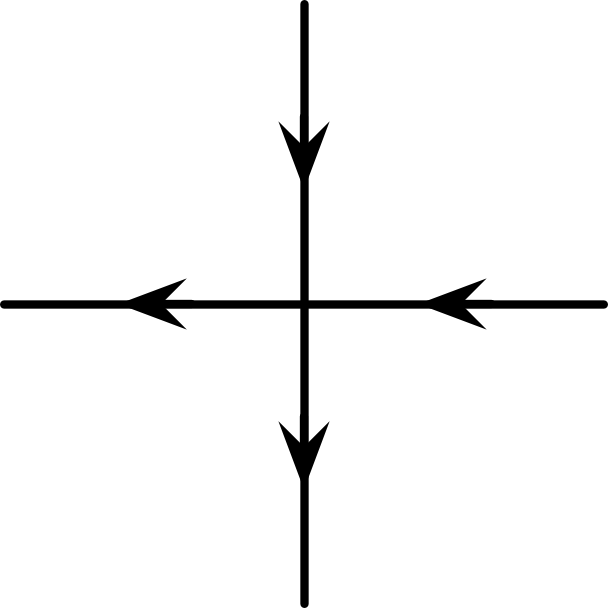}\caption{$2$}
\label{fig:orientations_2}
\end{subfigure}
\begin{subfigure}[b]{0.12\linewidth}
\centering\includegraphics[width=\linewidth]{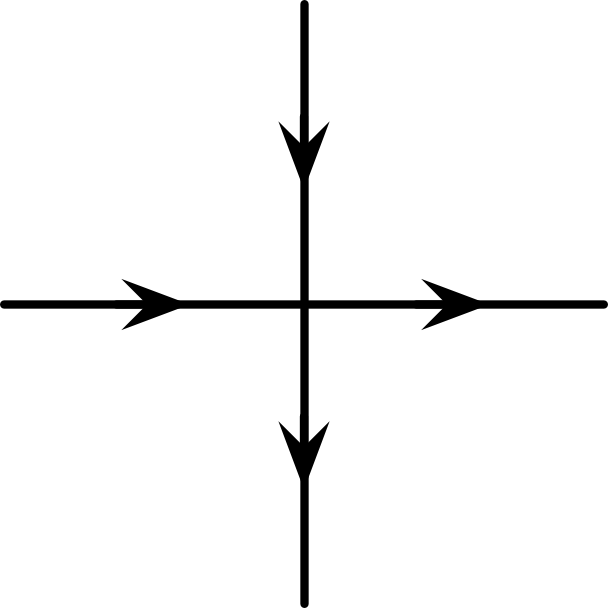}\caption{$3$}
\label{fig:orientations_3}
\end{subfigure}
\begin{subfigure}[b]{0.12\linewidth}
\centering\includegraphics[width=\linewidth]{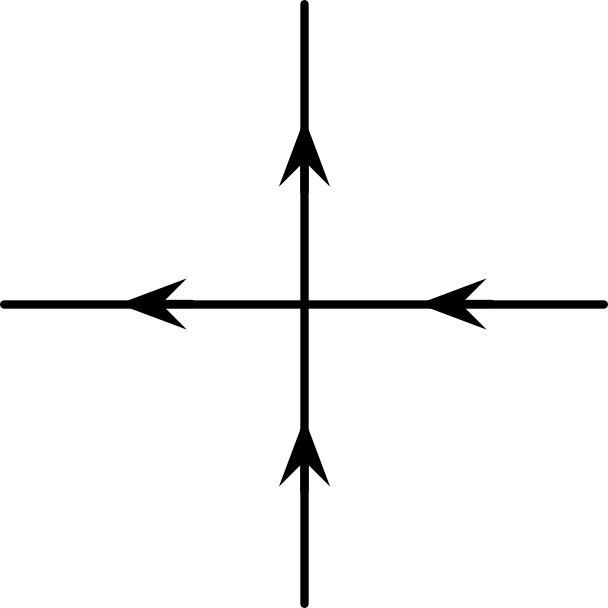}\caption{$4$}
\label{fig:orientations_4}
\end{subfigure}
\begin{subfigure}[b]{0.12\linewidth}
\centering\includegraphics[width=\linewidth]{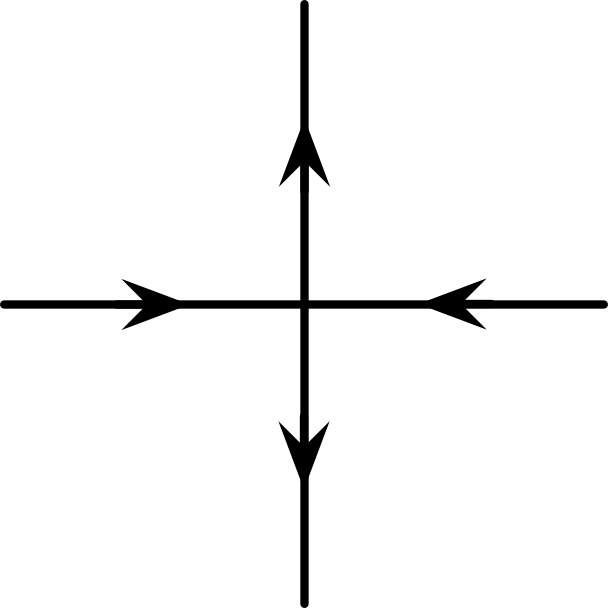}\caption{$5$}
\label{fig:orientations_5}
\end{subfigure}
\begin{subfigure}[b]{0.12\linewidth}
\centering\includegraphics[width=\linewidth]{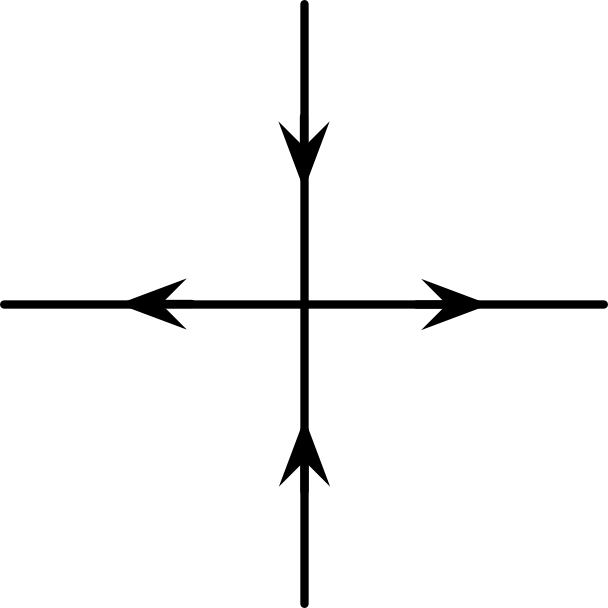}\caption{$6$}
\label{fig:orientations_6}
\end{subfigure}
\begin{subfigure}[b]{0.12\linewidth}
\centering\includegraphics[width=\linewidth]{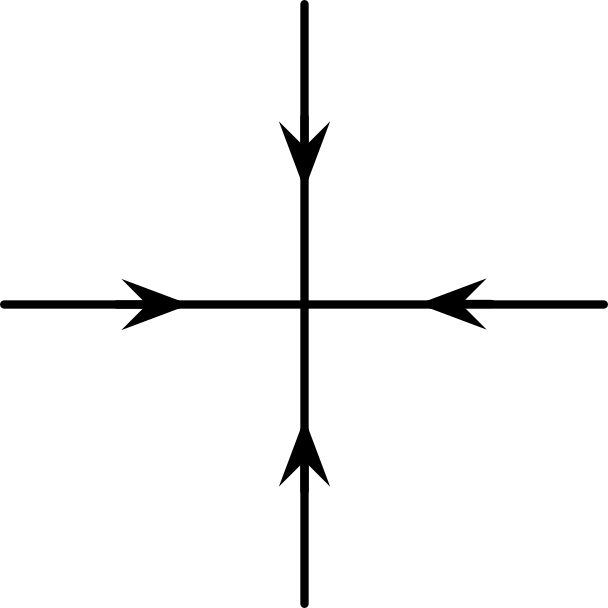}\caption{$7$}
\label{fig:orientations_7}
\end{subfigure}
\begin{subfigure}[b]{0.12\linewidth}
\centering\includegraphics[width=\linewidth]{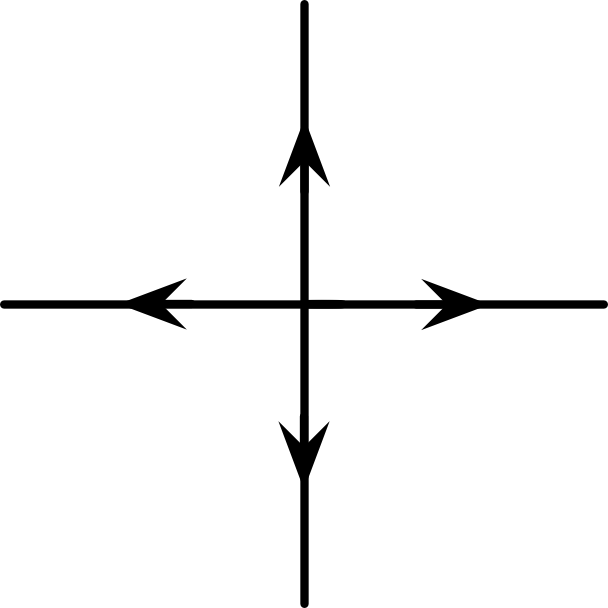}\caption{$8$}
\label{fig:orientations_8}
\end{subfigure}
\caption{Valid configurations of the eight-vertex model.}\label{fig:orientations}
\end{figure}
%\renewcommand{\thesubfigure}{\alph{subfigure}}
%\captionsetup[subfigure]{labelformat=parens}

Classically, the eight-vertex model
is defined by statistical physicists on a square lattice region where each vertex of the lattice is connected by an edge to four nearest neighbors.
%There are eight permitted types of local configurations around a vertex---hence the name eight-vertex model (see \figref{fig:orientations}).
In general, the eight configurations 1 to 8 in \figref{fig:orientations} are associated with eight possible weights $w_1, \ldots, w_8$. Denote the set of
these eight local configurations by $S_\textsc{8V}$.
By physical considerations, the total weight of a state remains unchanged
if  all arrows are flipped,
assuming there is no external electric field.
In this case we write
$w_1 = w_2 = a$, $w_3 = w_4= b$, $w_5 = w_6 = c$, and $w_7 = w_8 = d$.
This complementary invariance is known as the \emph{arrow reversal symmetry} or the \emph{zero field assumption}.

%If only six local arrangements 1 to 6 are permitted around a vertex (i.e. $d=0$), then the configurations are \emph{Eulerian orientations} of the underlying 4-regular graph (where each vertex has in-degree two and out-degree two).
%This is called the \emph{six-vertex model} which is the antecedent of the eight-vertex model.
%The latter was first introduced in 1970 by Sutherland~\cite{doi:10.1063/1.1665111}, and Fan and Wu~\cite{PhysRevB.2.723} as a generalization of the six-vertex model for certain
%% JYC: was just "its"
% more desirable properties on the square lattice. However in contrast to the six-vertex model which has been ``exactly solved'' (in the physics sense, a good understanding in the thermodynamic limit on the square lattice) under various parameter settings and external fields~\cite{PhysRev.162.162, PhysRevLett.18.1046, PhysRevLett.19.108, PhysRevLett.19.103, PhysRevB.2.723}, the eight-vertex model was ``exactly solved'' only in the zero-field case~\cite{PhysRevLett.26.832, BAXTER1972193}.
%Even in the zero-field setting, this model is already enormously expressive:
%its special case when $d=0$, the zero-field six-vertex model, has sub-models such as the ice ($a = b = c$), KDP, and Rys $F$ models; on the square lattice, some other important models such as the dimer and zero-field Ising models can be reduced to it~\cite{BAXTER1972193}. 
Even in the zero-field setting, this model is already enormously expressive.
The special case when $d=0$ is the six-vertex model, which itself has 
sub-models such as the ice ($a = b = c$), KDP, and Rys $F$ models; on the square lattice, some other important models such as the dimer and zero-field Ising models can be reduced to it~\cite{BAXTER1972193}.
Together with ferromagnetic Ising and monomer-dimer models,
the six-vertex and eight-vertex models are among the
most studied models in statistical physics\footnote{A search in
  \textit{Google Scholar}
for ``six- and eight-vertex models'' returns ``About 153,000 results''.}.
Beyond physics, Kuperberg
gave a simplified proof of the famous alternating-sign matrix (ASM)
conjecture in combinatorics using a connection to
 the six-vertex model~\cite{doi:10.1155/S1073792896000128}.
Recently, the six-vertex model played an important role in explicating the phase transition of the Potts model and the random cluster model on the square lattice~\cite{duminilcopin2016discontinuity, ray2019short}.
%After introduced in 1970 by Sutherland~\cite{doi:10.1063/1.1665111}, and Fan and Wu~\cite{PhysRevB.2.723}, the zero-field eight-vertex model was ``exactly solved'' by Baxter~\cite{PhysRevLett.26.832, BAXTER1972193} (in the physics sense, a good understanding in the thermodynamic limit on the square lattice). 
After the eight-vertex model was introduced in 1970 by Sutherland~\cite{doi:10.1063/1.1665111}, and Fan and Wu~\cite{PhysRevB.2.723}, Baxter~\cite{PhysRevLett.26.832, BAXTER1972193} achieved a good understanding of the zero-field case in the thermodynamic limit on the square lattice (in physics this understanding of the limiting case is called ``exactly solved'').
 
In this paper, we assume the arrow reversal symmetry and our algorithmic and complexity results further assume that $a, b, c, d \ge 0$ (as is the case in \emph{classical} physics), unless otherwise explicitly stated. 
%Given a  4-regular graph $G$ (not just the grid or even planar graph), we label
%four incident edges of each vertex
%from 1 to 4. 
% The \emph{partition function} of the eight-vertex model with parameters
%$(a, b, c, d)$ on  $G$
%is defined as
For any  4-regular graph $G$ (not just the grid or even planar graph,
however for plane graphs the edges are locally labeled from 1 to 4 cyclically), 
the \emph{partition function} of the eight-vertex model on $G$ with parameters
$(a, b, c, d)$
is defined as
\begin{equation}\label{Z-defn}
Z_{\textup{8V}}(G; a, b, c, d) = \sum_{\tau \in \mathcal{O}_{\bf e}(G)}a^{n_1 + n_2}b^{n_3 + n_4}c^{n_5 + n_6}d^{n_7 + n_8},
\end{equation}
where $\mathcal{O}_{\bf e}(G)$ is the set of all even orientations of $G$,
and $n_i$ is the number of vertices in type  $i$  in $G$ ($1 \le i \le 8$,
locally depicted as in     
Figure~\ref{fig:orientations}) 
under an even orientation $\tau \in \mathcal{O}_{\bf e}(G)$.

In terms of exact complexity, a dichotomy is given for the eight-vertex model on general 4-regular graphs for all eight (possibly complex) parameters~\cite{DBLP:journals/corr/CaiF17}.
This is studied in the context of a classification program for the complexity of counting problems~\cite{cai_chen_2017}, where the eight-vertex model serves as an important basic case for Holant problems defined by not necessarily symmetric constraint functions.
It is shown that
every setting is either P-time computable (and some are surprising) or \#P-hard.  
However, most cases for (exact) P-time tractability 
are due to nontrivial cancellations.
In our setting where $a, b, c, d$ are nonnegative real numbers, the problem of computing the partition function of the eight-vertex model is \#P-hard unless: (1) $a = b = c = d$ (this is equivalent to the unweighted case); 
(2) three of $a, b, c, d$ are zero; %at least three of $a, b, c, d$ are zero; 
or (3) two of $a, b, c, d$ are zero and the other two are equal.
The full classification of the exact complexity for the eight-vertex model on planar graphs is still open, but in the full version of this paper we will show that in our setting where $a, b, c, d$ are nonnegative, the problem that is \#P-hard on general graphs remains \#P-hard on planar graphs except in the following cases where it becomes P-time computable: (1) $a^2 + b^2 = c^2 + d^2$ or (2) one of $a, b$ is zero and one of $c, d$ is zero.

Recently, the approximate complexity of counting and sampling of the eight-vertex model (and its special case, the six-vertex model) has been studied~\cite{Greenberg2010, liu:LIPIcs:2018:9456, doi:10.1137/1.9781611975482.136, DBLP:journals/corr/abs-1904-01495, DBLP:journals/corr/abs-1811-03126, DBLP:journals/corr/abs-1904-10493}. Interestingly, these results conform to the phase transition phenomenon in physics. In order to state the previous results and present our work, we adopt the following notations assuming $a, b, c, d \ge 0$.
\begin{itemize}
%\item
%$\mathcal{A} = \{(a,b,c,d) \; | \; a \le b+c+d\}$, $\mathcal{B} = \{(a,b,c,d) \; | \; b \le a+c+d\}$, $\mathcal{C} = \{(a,b,c,d) \; | \; c \le a+b+d\}$, $\mathcal{D} = \{(a,b,c,d) \; | \; d \le a+b+c\}$;
\item
$\mathcal{X} = \{\; (a,b,c,d) \; | \; a \le b+c+d,\;\; b \le a+c+d,\;\; c \le b+c+d,\;\; d \le a+b+c\}$;
%\item
%$\mathcal{AD} = \{(a,b,c,d) \; | \; a+d \le b+c\}$, $\mathcal{BD} = \{(a,b,c,d) \; | \; b+d \le a+c\}$, $\mathcal{CD} = \{(a,b,c,d) \; | \; c+d \le a+b\}$;
\item
$\mathcal{Y} = \{\; (a,b,c,d) \; | \; a+d \le b+c,\;\; b+d \le a+c,\;\; c+d \le a+b\}$;
%\item
%$\mathcal{F}_{\le^2} = \{(a,b,c,d) \; | \;
%a^2 \le b^2 + c^2 + d^2, \;\;b^2 \le a^2 + c^2 + d^2,\;\; c^2 \le a^2 + b^2 + d^2,\;\;d^2 \le a^2 + b^2 + c^2\}$.
\item
$\mathcal{Z} = \{\; (a,b,c,d) \; | \;
a^2 \le b^2 + c^2 + d^2, \;\;b^2 \le a^2 + c^2 + d^2,\;\; c^2 \le a^2 + b^2 + d^2,\;\;d^2 \le a^2 + b^2 + c^2\}$.
\end{itemize}
\begin{remark}
$\mathcal{Y} \subset \mathcal{X}$ and
$\mathcal{Z} \subset \mathcal{X}$.
\end{remark}

Physicists have shown an \emph{order-disorder phase transition} for the eight-vertex model on the square lattice between parameter settings outside $\mathcal{X}$ and those inside (see Baxter's book~\cite{Baxter:book} for more details).
 Physicists  call $\mathcal{X}$  the \emph{disordered phase},
and its complement $\overline{\mathcal{X}}$, 
which consists of 4 disjoint regions in which one of
$(a, b, c, d)$ dominates, the \emph{ordered phases}.
In \cite{DBLP:journals/corr/abs-1811-03126} and \cite{DBLP:journals/corr/abs-1904-10493}, it was shown that: (1) approximating the partition function of the eight-vertex model on general 4-regular graphs outside $\mathcal{X}$ is NP-hard, (2) approximating the partition function of the eight-vertex model on general 4-regular graphs outside $\mathcal{Y}$ is at least as hard as approximately counting perfect matchings on general graphs (in short we say it is \#PM-hard), (3) there is an FPRAS\footnote{Suppose $f: \Sigma^* \rightarrow \mathbb{R}$ is a function mapping problem instances to real numbers. A \textit{fully polynomial randomized approximation scheme (FPRAS)} \cite{Karp:1983:MAE:1382437.1382804} for a problem is a randomized algorithm that takes as input an instance $x$ and $\varepsilon > 0$, running in time polynomial in the length $|x|$ and $\varepsilon^{-1}$, and outputs a number $Y$ (a random variable) such that
\(\operatorname{Pr}\left[(1 - \varepsilon)f(x) \le Y \le (1 + \varepsilon)f(x)\right] \ge \frac{3}{4}.\)} for general 4-regular graphs in $\mathcal{Y} \bigcap \mathcal{Z}$, and (4) there is an FPRAS for planar 4-regular graphs in a subregion of $\overline{\mathcal{Y}} \bigcap \mathcal{Z}$.
Note that all previous positive results are confined within $\mathcal{Z}$, in particular within the disordered phase $\mathcal{X}$.
See \figref{fig:complexity_landscape}.

\begin{figure}[h!]
\centering\includegraphics[width=0.7\linewidth]{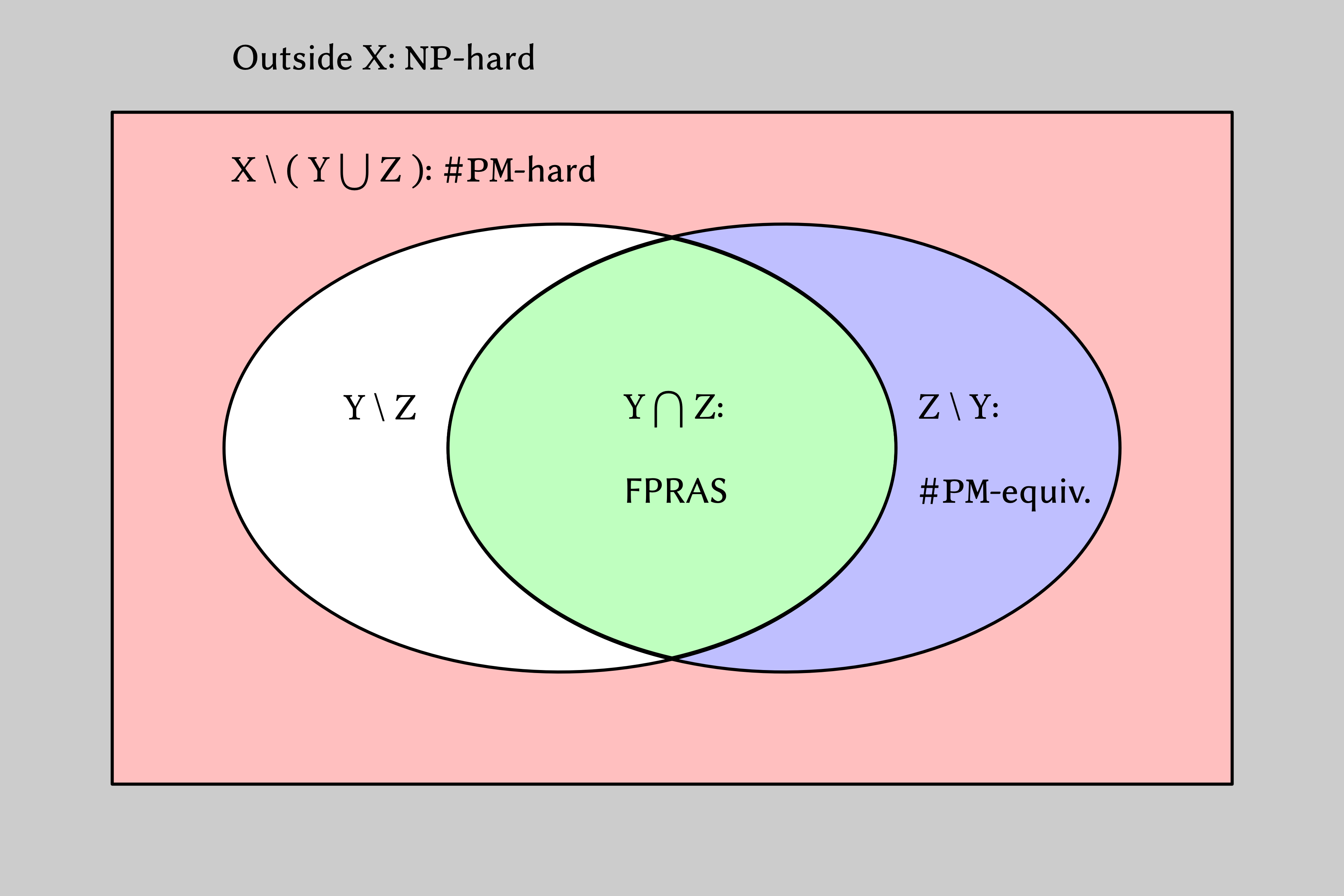}
\caption{A Venn diagram of the approximation complexity of the eight-vertex model on general 4-regular graphs.}
\label{fig:complexity_landscape}
\end{figure}

Previous FPRAS results in \cite{DBLP:journals/corr/abs-1811-03126} are based on the method of \emph{Markov chain Monte Carlo (MCMC)}.
A nice upper bound on the mixing time of a specific Markov chain can be achieved only in $\mathcal{Z}$ (which is a subregion of the disordered phase $\mathcal{X})$ using a \emph{canonical path argument}.
The canonical path argument was introduced for perfect matchings~\cite{doi:10.1137/0222066} and extends well to problems which are believed to have strong connections with perfect matchings~\cite{DBLP:journals/corr/abs-1301-2880}. 
%As it was proved in \cite{DBLP:journals/corr/abs-1904-10493} that computing $Z_{\textup{8V}}(a, b, c, d)$ for any $(a, b, c, d) \in \mathcal{F}_{\le^2}$ can be reduced to the problem of counting perfect matchings by expressing the local constraints $(a, b, c, d)$ using ``matchgates'' whereas such expression fails for $(a, b, c, d) \not\in \mathcal{F}_{\le^2}$, it seems that there is intrinsic difficulty in bounding the mixing time using MCMC (or at least canonical path argument) directly.
%Moreover, in the ordered phases of the eight-vertex model ($\overline{\mathcal{A}}$, $\overline{\mathcal{B}}$, $\overline{\mathcal{C}}$, or $\overline{\mathcal{D}}$), torpid mixing results of natural Markov chains were established on grid/torus graphs~\cite{Greenberg2010, liu:LIPIcs:2018:9456, DBLP:journals/corr/abs-1904-01495}. 
We proved in \cite{DBLP:journals/corr/abs-1904-10493} that computing $Z_{\textup{8V}}(a, b, c, d)$ for any $(a, b, c, d) \in \mathcal{Z}$ can be reduced to the problem of counting perfect matchings by expressing the local constraints $(a, b, c, d)$ using ``matchgates'' whereas such expression provably fails for $(a, b, c, d) \not\in \mathcal{Z}$.
% suggesting an intrinsic barrier.
Moreover, in the ordered phases of the eight-vertex model 
(the four disjoint parts
of $\overline{\mathcal{X}}$), torpid mixing results of natural Markov chains were established on grid/torus graphs~\cite{Greenberg2010, liu:LIPIcs:2018:9456, DBLP:journals/corr/abs-1904-01495}.

In this paper, we give the first FPRAS for $Z_{\textup{8V}}(a, b, c, d)$ for $(a, b, c, d)$ outside $\mathcal{Z}$ on planar and on bipartite graphs.
We introduce a special edge-2-coloring model, called the \emph{even-coloring model} (\secref{sec:even-coloring}), and simultaneously set up two different kinds of relations between the partition functions of the eight-vertex model and the even-coloring model.
The first kind of relations exploit the property of planar and of bipartite graphs, and are one-to-one weight-preserving mappings between the states of the eight-vertex model $\mathcal{O}_{\bf e}(G)$ and the states of the even-coloring model $\mathcal{C}_{\bf e}(G)$; the second is by the method of \emph{holographic transformation} introduced by Valiant~\cite{Valiant:2008:HA:1350684.1350697} which can be thought of as an (exponentially many)-to-(exponentially many) map between the two state spaces (\secref{sec:holo_trans}).
This, magically, allows us to identify the partition functions of the eight-vertex model on the same graph under \emph{totally different} parameter settings.
Interestingly, we show that these maps and their compositions among different parameter settings under which the partition function is preserved have \emph{group structures}.
For planar graphs, this group is isomorphic to the \emph{symmetry group} $S_3$ on three elements (see \secref{sec:planar}); for bipartite graphs, this group is isomorphic to the \emph{dihedral group} $D_6$ of order 12
(the symmetry group of a regular hexagon, see \secref{sec:bipartite}).

%Therefore, although the Markov chain on a graph under certain parameter settings outside $\mathcal{F}_{\le^2}$ resists a canonical path argument and thus might not be rapidly mixing, 
Therefore, although the Markov chain on a graph under certain parameter settings outside $\mathcal{Z}$ is not rapidly mixing,
after ``mixing up'' the state space using a combination of two maps, the Markov chain turns out to be rapidly mixing.
Indeed, as a consequence, this ``indirect'' MCMC leads to FPRAS for new regions in the disordered phase $\mathcal{X}$ and, for the first time, in the ordered phases $\overline{\mathcal{X}}$ for planar graphs and for bipartite graphs.

\begin{theorem}\label{thm:main_planar}
Let $G$ be a 4-regular plane graph. There is an FPRAS for $Z_{\textup{8V}}(G; a, b, c, d)$ for $(a, b, c, d)$ in a subregion of $\mathcal{X} \bigcap \overline{\mathcal{Y}} \bigcap \overline{\mathcal{Z}}$ and in a subregion of $\overline{\mathcal{X}}$.
\end{theorem}

We have proved that
%recall that 
on general 4-regular graphs, approximating the eight-vertex model in $\mathcal{X} \bigcap \overline{\mathcal{Y}}$ is \#PM-hard~\cite{DBLP:journals/corr/abs-1904-10493} and in $\overline{\mathcal{X}}$ is NP-hard~\cite{DBLP:journals/corr/abs-1811-03126}.
Therefore, we have found a family of problems (with parameters ranging in a region of
parameter space) having the following
%three provable properties:
provable properties:
For the eight-vertex model in the subregion of $\overline{\mathcal{X}}$ given in \thmref{thm:main_planar} (described more explicitly in \corref{cor:planar}), computing $Z_{\textup{8V}}(a, b, c, d)$ is
\begin{enumerate}
\item
% (1) \#P-complete in exact computation even on planar graphs (proof delayed to the full version of this paper), 
%(2)
NP-hard to approximate on general 4-regular graphs~\cite{DBLP:journals/corr/abs-1811-03126}, and
\item
% (3) 
has an FPRAS on planar 4-regular graphs (this paper).
\end{enumerate}
This separation of complexity for general and for planar graphs
should be compared and contrasted with the \emph{FKT algorithm}~\cite{doi:10.1080/14786436108243366, KASTELEYN19611209, kasteleyn_book} for exact counting of perfect
matchings, but here for approximate counting.
Previously the combined results of \cite{10.1145/2785964}
and \cite{doi:10.1002/rsa.20560} proved a similar result for
$k$-colorings for general versus planar graphs.
We note that it was shown in \cite{GOLDBERG2015330} that approximating the partition functions of many two-state spin systems remain NP-hard on planar graphs.
A similar result was shown in \cite{Goldberg2012} for approximating the Tutte polynomial $T(G; x, y)$ in a large portion of the $(x,y)$ plane.
We can also prove that for the subregion, $Z_{\textup{8V}}(a, b, c, d)$ is
\#P-complete in exact computation even on planar graphs 
(we will include the proof in the extended version of this paper).
%(proof delayed to the full version of this paper).

\begin{theorem}\label{thm:main_bipartite}
Let $G$ be a 4-regular bipartite graph. There is an FPRAS for $Z_{\textup{8V}}(G; a, b, c, d)$ for $(a, b, c, d)$ in a subregion of $\mathcal{X} \bigcap \overline{\mathcal{Y}}$ and in a subregion of $\overline{\mathcal{X}}$.
\end{theorem}

Note that the subregions mentioned in \thmref{thm:main_bipartite} are disjoint from those mentioned in \thmref{thm:main_planar}.
We have proved that
%recall that 
on general 4-regular graphs, approximating the eight-vertex model in $\mathcal{X} \bigcap \overline{\mathcal{Y}}$ is \#PM-hard~\cite{DBLP:journals/corr/abs-1904-10493} and in $\overline{\mathcal{X}}$ is NP-hard~\cite{DBLP:journals/corr/abs-1811-03126}.
Therefore, we have found a family of problems (with parameters ranging in a region of
parameter space) having the following
provable properties:
For the eight-vertex model in the subregion of $\overline{\mathcal{X}}$ given in \thmref{thm:main_bipartite} (described more explicitly in \corref{cor:bipartite}), computing $Z_{\textup{8V}}(a, b, c, d)$ is
%\begin{enumerate}
%\item  NP-hard in approximate computation on general 4-regular graphs~\cite{DBLP:journals/corr/abs-1811-03126}, and 
%\item admits an FPRAS in approximate computation on bipartite 4-regular graphs (this paper).
%\end{enumerate}
\begin{enumerate}
\item
NP-hard to approximate on general 4-regular graphs~\cite{DBLP:journals/corr/abs-1811-03126}, and
\item
has an FPRAS on bipartite 4-regular graphs (this paper).
\end{enumerate}
%We also note that the approximation in $\overline{\mathcal{A}}$, $\overline{\mathcal{B}}$, and $\overline{\mathcal{C}}$ is proved to be NP-hard even on bipartite graphs~\cite{DBLP:journals/corr/abs-1811-03126}.
Previously the only problem having similar properties is the antiferromagnetic Ising model on general versus bipartite graphs~\cite{10.2307/24519110, CAI2016690}.
We note that the problem of counting independent sets on bipartite graphs (\#BIS) 
%is a complete problem for a complexity class believed to be in the middle of $P$ and $NP$~\cite{}. 
is considered a canonical counting problem of intermediate approximation complexity~\cite{Dyer2004}. It is conjectured that \#BIS neither has an FPRAS nor
% is as hard as \#SAT to approximate.
is NP-hard to approximate.
Many 2-spin systems on bipartite graphs are only known to be
 \#BIS-hard or \#BIS-equivalent (below NP-hard)
  to approximate~\cite{CAI2016690}.
%It was shown that \#BIS-hard to approximate the partition function of many 2-spin system on bipartite graphs~\cite{CAI2016690}.
%

To summarize, this paper establishes the eight-vertex model as
the first problem with the provable property that
while NP-hard to approximate on general graphs (even \#P-hard for  
planar graphs in exact complexity),
it possesses FPRAS on both bipartite graphs and
planar graphs in substantial regions of its parameter space.

A key property we use in our proof is that the dual of any planar 4-regular graph is bipartite~\cite{WELSH1969375}, and hence 2-colorable.
Since torus graphs $(\mathbb{Z}/m\mathbb{Z}) \times (\mathbb{Z}/n\mathbb{Z})$ with even $m, n$ also have this property and are bipartite, the results in
 \thmref{thm:main_planar}  and \thmref{thm:main_bipartite}  
%%%% \thmref{thm:planar} and \thmref{thm:bipartite} 
also hold in torus graphs (with even side lengths).

Finally, we note that the techniques introduced in this paper have other applications.
First,
the maps between partition functions under different parameter settings established in this paper are not only useful for giving approximation algorithms. The same maps are useful in our understanding of the exact computational complexity of the eight-vertex model on planar graphs. 
%\item
%The general paradigm can be applied to the study of the eight-vertex model on other classes of graphs in additional to planar/bipartite/torus graphs.
%In particular, the methodology can be readily extended to any class of graphs with a ``canonical'' even orientation where every vertex has the same weight ($a$, $b$, $c$, or $d$). 
Second,
the techniques are useful for studying the partition functions of other edge-orientation problems (e.g. other vertex models in statistical physics) or edge-coloring problems (e.g. Holant problems) on planar/bipartite graphs.

\bigskip

\section{The even-coloring model}\label{sec:even-coloring}
We introduce the following edge-2-coloring model on 4-regular graphs called the \emph{even-coloring model}: a valid configuration of this model assigns
either  \emph{green} or \emph{red} to every edge such that 
the number of green edges incident to any vertex is even (zero, two, or four).
Similar to the eight-vertex model, there are also eight valid local configurations around a vertex (shown in \figref{fig:even-colorings}), and configurations 1 to 8 in \figref{fig:even-colorings} can be associated with weights $w'_1, \dots, w'_8$ respectively.
Denote the set of
these eight local configurations by $S_\textsc{EC}$.

\begin{figure}[h!]
\centering
\begin{subfigure}[b]{0.12\linewidth}
\centering\includegraphics[width=\linewidth]{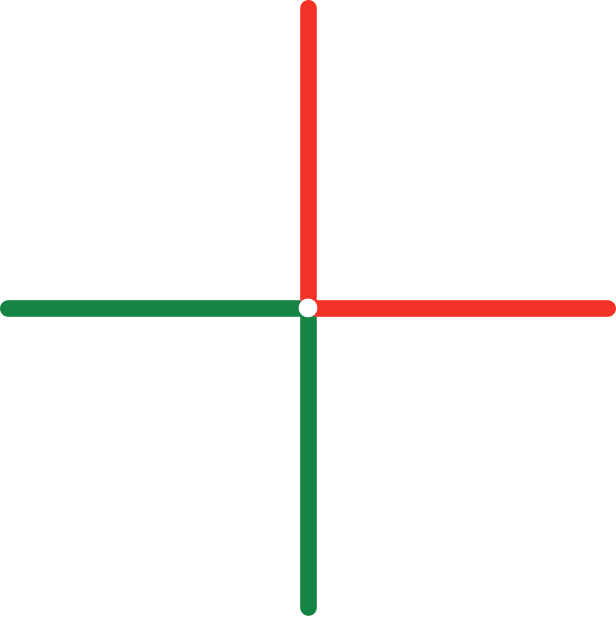}\caption{$1$}
\label{fig:coloring_1}
\end{subfigure}
\begin{subfigure}[b]{0.12\linewidth}
\centering\includegraphics[width=\linewidth]{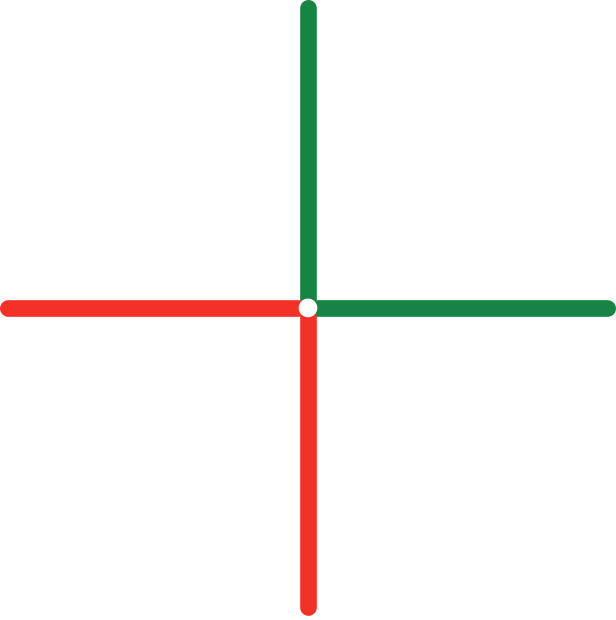}\caption{$2$}
\label{fig:coloring_2}
\end{subfigure}
\begin{subfigure}[b]{0.12\linewidth}
\centering\includegraphics[width=\linewidth]{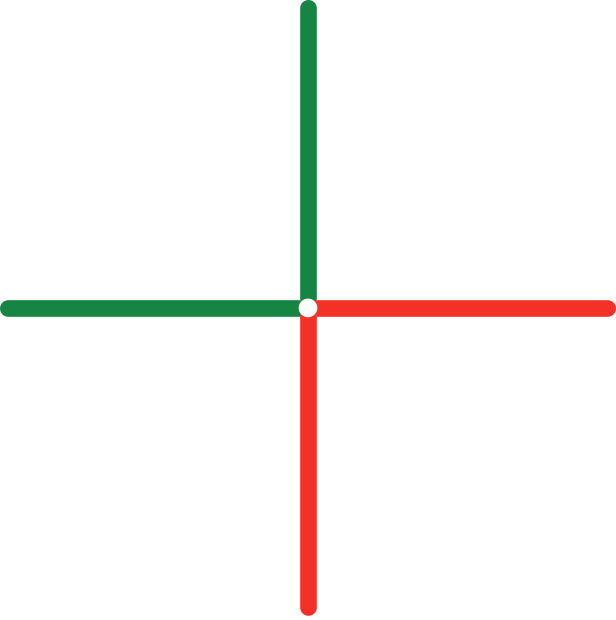}\caption{$3$}
\label{fig:coloring_3}
\end{subfigure}
\begin{subfigure}[b]{0.12\linewidth}
\centering\includegraphics[width=\linewidth]{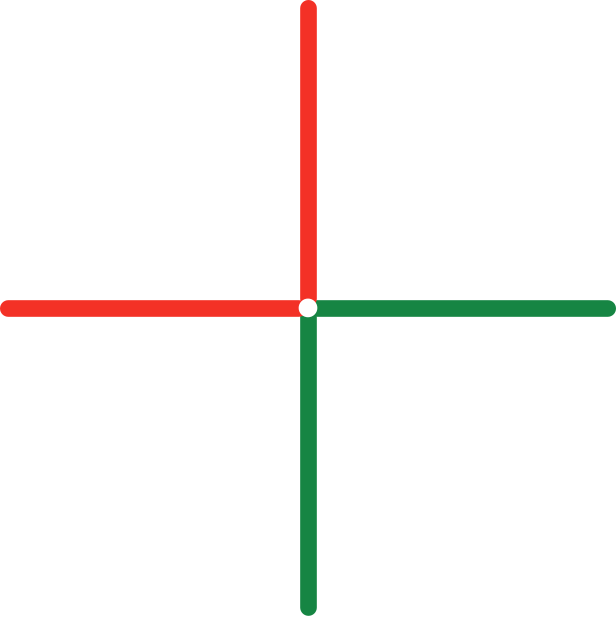}\caption{$4$}
\label{fig:coloring_4}
\end{subfigure}
\begin{subfigure}[b]{0.12\linewidth}
\centering\includegraphics[width=\linewidth]{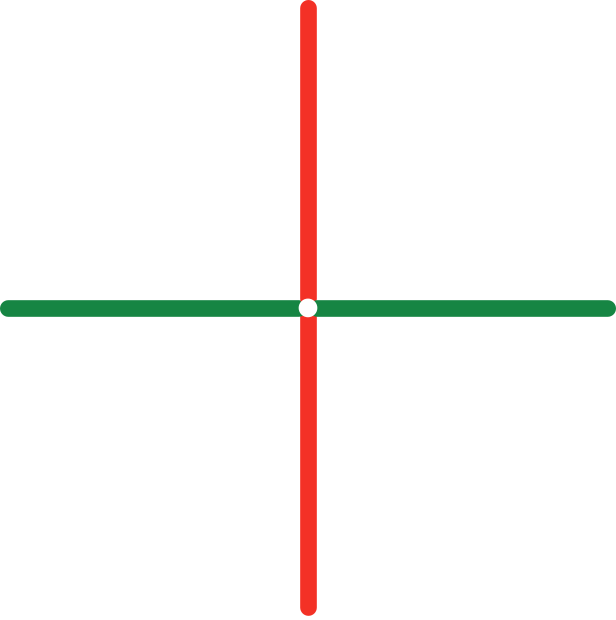}\caption{$5$}
\label{fig:coloring_5}
\end{subfigure}
\begin{subfigure}[b]{0.12\linewidth}
\centering\includegraphics[width=\linewidth]{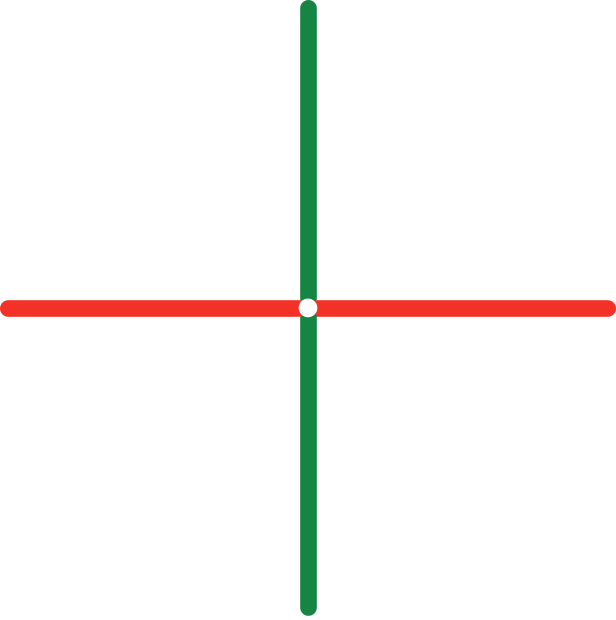}\caption{$6$}
\label{fig:coloring_6}
\end{subfigure}
\begin{subfigure}[b]{0.12\linewidth}
\centering\includegraphics[width=\linewidth]{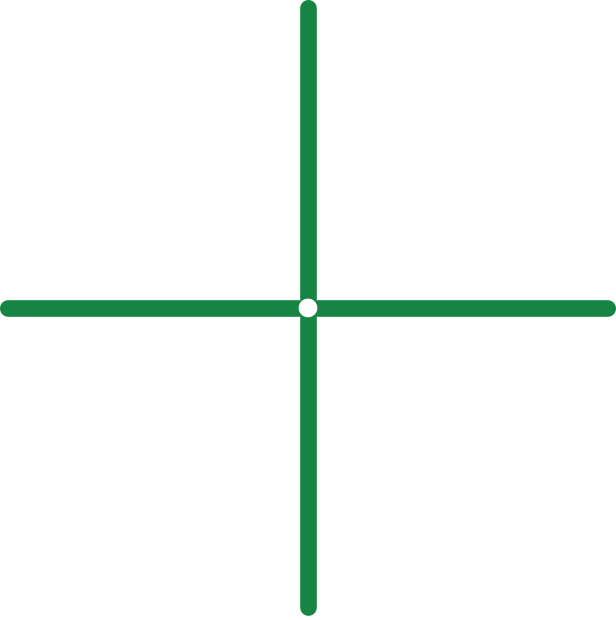}\caption{$7$}
\label{fig:coloring_7}
\end{subfigure}
\begin{subfigure}[b]{0.12\linewidth}
\centering\includegraphics[width=\linewidth]{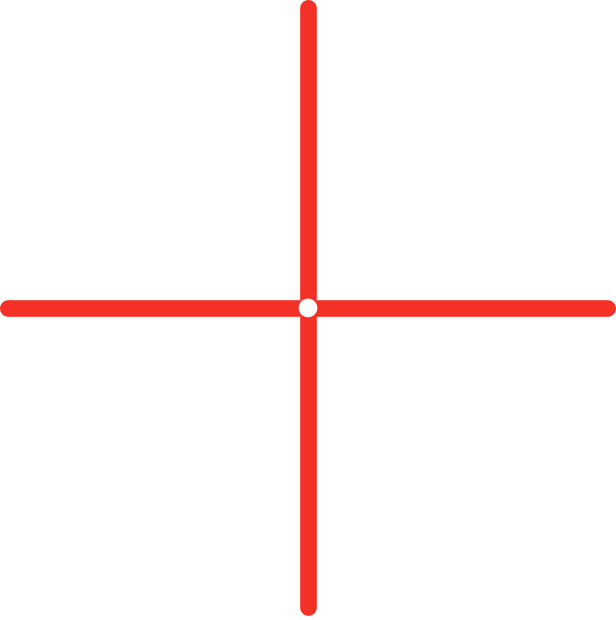}\caption{$8$}
\label{fig:coloring_8}
\end{subfigure}
\caption{Valid configurations of the even-coloring model.}\label{fig:even-colorings}
\end{figure}

To set up a correspondence between the even-coloring model and the eight-vertex model satisfying arrow reversal symmetry, we consider the even-coloring model 
with weights that satisfy the \emph{color reversal symmetry}. That is, the weight of a local configuration at a vertex remains unchanged if the color on every incident edge is changed.
In this case we write
$w'_1 = w'_2 = w$, $w'_3 = w'_4= x$, $w'_5 = w'_6 = y$, and $w'_7 = w'_8 = z$. 
%In this paper, we make this assumption
%and further assume that
% $w, x, y, z \ge 0$.
%%% maybe change above sentence to:
%%% In this paper, our algorithmic results assume that $w, x, y, z \ge 0$. 
 Given a  4-regular graph $G$, we label
four incident edges of each vertex
from 1 to 4.
The \emph{partition function} of the even-coloring model with parameters
$(w, x, y, z)$ on  $G$
is defined as
\begin{equation}\label{defn:even-coloring}
Z_{\textup{EC}}(G; w, x, y, z) = \sum_{\varsigma \in \mathcal{C}_{\bf e}(G)}w^{n_1 + n_2}x^{n_3 + n_4}y^{n_5 + n_6}z^{n_7 + n_8},
\end{equation}
where $\mathcal{C}_{\bf e}(G)$ is the set of all even colorings of $G$,
and $n_i$ is the number of vertices in type  $i$  in $G$ ($1 \le i \le 8$,
locally depicted as in     
\figref{fig:even-colorings}) 
under the even-coloring $\varsigma \in \mathcal{C}_{\bf e}(G)$.

%If we model a green-red edge coloring by a 0-1 assignment 
%to the edges such that  an edge $e$ is assigned $0$ if it is colored green
%and assigned  $1$ if it is colored red,
% then the partition function of the even-coloring model $Z_{\textup{EC}}(G; w, x, y, z)$ is exactly the
%value of the Holant problem $\operatorname{Holant}\left(G; \left[\begin{smallmatrix} w & 0 & 0 & x \\ 0 & y & z & 0 \\ 0 & z & y & 0 \\ x & 0 & 0 & w \end{smallmatrix}\right]\right)$.

\bigskip

\section{Holographic Transformation}\label{sec:holo_trans}
%\documentclass[paper]{subfiles}

%\subsection{The eight-vertex model as a Holant problem}

Given a 4-regular graph $G = (V, E)$,
the \emph{edge-vertex incidence graph}  $G' = (U_E, U_V, E')$
is a bipartite graph where
$(u_e, u_v) \in U_E \times U_V$ is an edge in $E'$ iff 
$e \in E$ in $G$  is incident to $v \in V$.
We model an orientation ($w \rightarrow v$)
 on an edge $e = \left\{w, v\right\} \in E$
 from $w$ into $v$ in $G$ by assigning
 $1$  to $(u_e, u_w) \in E'$  and  $0$ to $(u_e, u_v) \in E'$
in $G'$.
A configuration of the eight-vertex model on $G$
is  a
\emph{0-1 labeling} on $G'$,
namely $\sigma: E' \rightarrow \{0, 1\}$,
where for  each  $u_e \in U_E$ its two incident edges are
assigned 01 or 10, and  for  each $u_v \in U_V$ the  sum of
 values $\sum_{i=1}^4 \sigma(e_i) \equiv 0  \pmod 2$,
over the four incident edges of $u_v$.
Thus 
we model the even orientation rule of $G$ on all $v \in V$ by requiring ``two-0-two-1/four-0/four-1'' locally at
each vertex $u_v \in U_V$.

The ``one-0-one-1'' requirement on the two edges incident to a vertex in $U_E$ is a binary {\sc Disequality} constraint, denoted by $(\neq_2)$.
The values of a 4-ary \emph{constraint function} $f$  can be listed in a matrix $M(f) = \left[\begin{smallmatrix} f_{0000} & f_{0010} & f_{0001} & f_{0011} \\ f_{0100} & f_{0110} & f_{0101} & f_{0111} \\ f_{1000} & f_{1010} & f_{1001} & f_{1011} \\ f_{1100} & f_{1110} & f_{1101} & f_{1111}\end{smallmatrix}\right]$,
called the \emph{constraint matrix} of $f$. For the eight-vertex model 
satisfying the even orientation rule and arrow reversal symmetry, the constraint function $f$ at every vertex $v \in U_V$ in $G'$ 
has the form $M(f) = \left[\begin{smallmatrix} d & 0 & 0 & a \\ 0 & b & c & 0 \\ 0 & c & b & 0 \\ a & 0 & 0 & d \end{smallmatrix}\right]$, if we locally index the left, down, right, and up edges incident to $v$ by 1, 2, 3, and 4, respectively according to \figref{fig:orientations}.
Thus computing the partition function $Z_{\textup{8V}}(G; a, b, c, d)$ is equivalent to evaluating
%(the Holant sum in the framework for Holant  problems)
\[\sum_{\sigma:E'\rightarrow\left\{0,1\right\}}\prod_{u\in U_E}(\neq_2)\left(\sigma |_{E'(u)}\right) \prod_{u\in U_V}f\left(\sigma |_{E'(u)}\right),\]
where $E'(u)$ denotes the incident edges of $u \in U_E \cup U_V$.
In fact, in this way we express the partition function of the eight-vertex model as the Holant sum in the framework for Holant problems:
\[Z_{\textup{8V}}(G; a, b, c, d) = \textup{Holant}\left(G'; \neq_2 |\ f\right)\]
where we use $\textup{Holant}(H; g\ |\ f)$ to denote the Holant sum 
$\sum_{\sigma:E\rightarrow\left\{0,1\right\}}\prod_{u\in U}g\left(\sigma |_{E(u)}\right) \prod_{u\in V}f\left(\sigma |_{E(u)}\right)$
on a bipartite graph $H = (U, V, E)$ for the Holant problem $\textup{Holant}(g\ |\ f)$.
Each vertex in $U$ (or $V$) is assigned the constraint function $g$ (or $f$, respectively).
The constraint function $g$ is written as a row vector, whereas the constraint function $f$ is written as a column vector, both as truth tables. (See \cite{cai_chen_2017} for more on Holant problems.)
The following proposition says that an invertible holographic transformation does not change the complexity of the Holant problem in the bipartite setting.
%Let $T$ be an invertible $2$-by-$2$ matrix. Let $d_1 = \operatorname{arity}(g)$ and $d_2 = \operatorname{arity}(f)$. Define $g' = g \left(T^{-1}\right)^{\otimes d_1}$ and $f' = T^{\otimes d_2} f$.
%\begin{proposition}[\cite{Valiant:2008:HA:1350684.1350697}]\label{prop:holo_trans}
%If $T \in \mathbb{C}^2$ is an invertible matrix, then for any bipartite graph $H$, $\textup{Holant}(H; g\ |\ f) = \textup{Holant}(H; g'\ |\ f')$.
%\end{proposition}
\begin{proposition}[\cite{Valiant:2008:HA:1350684.1350697}]\label{prop:holo_trans}
Suppose $T \in \mathbb{C}^2$ is an invertible matrix. Let $d_1 = \operatorname{arity}(g)$ and $d_2 = \operatorname{arity}(f)$. Define $g' = g \left(T^{-1}\right)^{\otimes d_1}$ and $f' = T^{\otimes d_2} f$.
Then for any bipartite graph $H$, $\textup{Holant}(H; g\ |\ f) = \textup{Holant}(H; g'\ |\ f')$.
\end{proposition}
%When every vertex in $G$ 
%has the same constraint function $f$ with $M(f) = \left[\begin{smallmatrix} d & & & a \\ & b & c & \\ & c & b & \\ a & & & d \end{smallmatrix}\right]$,
%we write the partition function $Z(a, b, c, d)$ as $Z(f)$,
% and denote by $Z(\mathcal{F})$ when
%each vertex is assigned some constraint function from a set $\mathcal{F}$
%consisting of constraint functions of this form.

%We denote $\textup{Holant}(G; f) = \textup{Holant}(G'; =_2 |\ f)$; this is equivalent to the usual definition.
We denote $\textup{Holant}(G; f) = \textup{Holant}(G'; =_2 |\ f)$.
For the even-coloring model, if we view a green-red edge coloring by a 0-1 assignment 
to the edges such that  an edge $e$ is assigned $0$ if it is colored green
and assigned  $1$ if it is colored red,
 then the partition function of the even-coloring model $Z_{\textup{EC}}(G; w, x, y, z)$ is exactly the
value of the Holant problem $\operatorname{Holant}\left(G; \left[\begin{smallmatrix} z & 0 & 0 & w \\ 0 & x & y & 0 \\ 0 & y & x & 0 \\ w & 0 & 0 & z \end{smallmatrix}\right]\right)$.

The following two lemmas show that
the eight-vertex model and the even-coloring model are connected via suitable
holographic transformations in unexpected ways as Holant problems.

\begin{lemma}\label{lem:pm-hard_holant}
Let $G$ be a 4-regular graph
%Then $Z_{\textup{8V}}(G; a, b, c, d) = \textup{Holant}\left(G; \frac{1}{2}\left[\begin{smallmatrix} a + b + c + d & 0 & 0 & - a + b + c - d \\ 0 & a - b + c - d & a +b - c - d & 0 \\ 0 & a +b - c - d & a - b + c - d & 0 \\ - a + b + c - d & 0 & 0 & a + b + c + d \end{smallmatrix}\right]\right)$.
and let $M_{Z} = %_\textup{\textsc{Planar}} = 
\frac{1}{2} \left[\begin{smallmatrix} -1 & 1 & 1 & -1 \\ 1 & -1 & 1 & -1 \\ 1 & 1 & -1 & -1 \\ 1 & 1 & 1 & 1 \end{smallmatrix}\right]$.
Then
$Z_{\textup{8V}}(G; a, b, c, d) = Z_{\textup{EC}}(G; w, x, y, z)$ where
$\left[\begin{smallmatrix} w\\ x\\ y\\ z \end{smallmatrix}\right]
= M_{Z}%_\textup{\textsc{Planar}}
\left[\begin{smallmatrix} a\\ b\\ c\\ d \end{smallmatrix}\right]$.
\end{lemma}
\begin{proof}
Using the binary disequality function $(\not=_2)$ for
the orientation of any edge, we can express  the partition function of the eight-vertex model $G$ as a Holant problem on its edge-vertex incidence graph $G'$,
\[Z_{\textup{8V}}(G; a, b, c, d) = \textup{Holant}\left(G'; \neq_2 |\ f\right),\]
where $f$ is the 4-ary signature with $M(f) = \left[\begin{smallmatrix} d & 0 & 0 & a \\ 0 & b & c & 0 \\ 0 & c & b & 0 \\ a & 0 & 0 & d \end{smallmatrix}\right]$.
Note that, writing the truth table of
$(\neq_2) = (0, 1, 1, 0)$ as a vector
and multiplied by a tensor power of the matrix $Z^{-1}$, where  
$Z = \frac{1}{\sqrt{2}}\left[\begin{smallmatrix} 1 & 1 \\ i & -i \end{smallmatrix}\right]$
we get
$(\neq_2) (Z^{-1})^{\otimes 2}
= (1, 0, 0, 1)$, which is exactly the
truth table of the binary equality function
$(=_2)$.  Then according to \propref{prop:holo_trans},
by the $Z$-transformation, we get 
%\begin{align*}
%\textup{Holant}\left(G'; \neq_2 | \left[\begin{smallmatrix} d & 0 & 0 & a \\ 0 & b & c & 0 \\ 0 & c & b & 0 \\ a & 0 & 0 & d \end{smallmatrix}\right]\right)
%& = \textup{Holant}\left(G'; \neq_2 \cdot \left( Z^{-1} \right)^{\otimes 2} |\ Z^{\otimes 4} \cdot \left[\begin{smallmatrix} d & 0 & 0 & a \\ 0 & b & c & 0 \\ 0 & c & b & 0 \\ a & 0 & 0 & d \end{smallmatrix}\right]\right) \\
%& = \textup{Holant}\left(G'; =_2 | \left[\begin{smallmatrix} a + b + c + d & 0 & 0 & - a + b + c - d \\ 0 & a - b + c - d & a +b - c - d & 0 \\ 0 & a +b - c - d & a - b + c - d & 0 \\ - a + b + c - d & 0 & 0 & a + b + c + d \end{smallmatrix}\right]\right)\\
%& = \textup{Holant}\left(G; \left[\begin{smallmatrix} a + b + c + d & 0 & 0 & - a + b + c - d \\ 0 & a - b + c - d & a +b - c - d & 0 \\ 0 & a +b - c - d & a - b + c - d & 0 \\ - a + b + c - d & 0 & 0 & a + b + c + d \end{smallmatrix}\right]\right).
%\end{align*}
\begin{align*}
\textup{Holant}\left(G'; \neq_2 |\ f\right)
& = \textup{Holant}\left(G'; \neq_2 \cdot \left( Z^{-1} \right)^{\otimes 2} |\ Z^{\otimes 4} \cdot f\right) \\
& = \textup{Holant}\left(G'; =_2 |\ Z^{\otimes 4} f\right)\\
& = \textup{Holant}\left(G; \ Z^{\otimes 4} f\right),
\end{align*}
and a direct calculation shows that $M(Z^{\otimes 4} f) = \frac{1}{2}
\left[\begin{smallmatrix} a + b + c + d & 0 & 0 & - a + b + c - d \\ 0 & a - b + c - d & a +b - c - d & 0 \\ 0 & a +b - c - d & a - b + c - d & 0 \\ - a + b + c - d & 0 & 0 & a + b + c + d \end{smallmatrix}\right]$.
\end{proof}

Readers are referred to \appref{app:holo_trans} for a more insightful
explanation on why the arity-4 constraint function $f$ is transformed to a \emph{real-valued} constraint function, under the \emph{complex-valued} $Z$-transformation.

%\subsection{Holographic transformation}
\begin{lemma}\label{lem:holo_trans}
%Let $G$ be a 4-regular graph.
%Let $f^* : \mathbb{R}^4 \rightarrow \mathbb{R}^4$ such that $f^*(a, b, c, d) = (a^*, b^*, c^*, d^*)$ where
%$\left\{\begin{smallmatrix}
%a^* &=& \frac{- a + b + c + d}{2} \\
%b^* &=& \frac{a - b + c + d}{2} \\
%c^* &=& \frac{a + b - c + d}{2} \\
%d^* &=& \frac{a + b + c - d}{2}
%\end{smallmatrix}\right.$.
%Then 
%\(Z_{\textup{8V}}(G; a, b, c, d) = Z_{\textup{EC}}(G; a^*, b^*, c^*, d^*).\)
Let $G$ be a 4-regular graph
and let $M_{HZ} = 
\frac{1}{2} \left[\begin{smallmatrix} -1 & 1 & 1 & 1 \\ 1 & -1 & 1 & 1 \\ 1 & 1 & -1 & 1 \\ 1 & 1 & 1 & -1 \end{smallmatrix}\right]$.
Then
$Z_{\textup{8V}}(G; a, b, c, d) = Z_{\textup{EC}}(G; w, x, y, z)$ where
$\left[\begin{smallmatrix} w\\ x\\ y\\ z \end{smallmatrix}\right]
= M_{HZ}%_\textup{\textsc{Planar}}
\left[\begin{smallmatrix} a\\ b\\ c\\ d \end{smallmatrix}\right]$.
\end{lemma}
\begin{proof}
%Using the binary disequality function $(\not=_2)$ for
%the orientation of any edge, we can express  the partition function of the eight-vertex model $G$ as a Holant problem on its edge-vertex incidence graph $G'$,
%\[Z_{\textup{8V}}(G; a, b, c, d) = \textup{Holant}\left(G'; \neq_2 |\ f\right),\]
%where $f$ is the 4-ary constraint function with $M(f) = \left[\begin{smallmatrix} d & 0 & 0 & a \\ 0 & b & c & 0 \\ 0 & c & b & 0 \\ a & 0 & 0 & d \end{smallmatrix}\right]$.
For the eight-vertex model as a Holant problem $\textup{Holant}\left(G'; \neq_2 |\ f\right)$,
we perform a holographic transformation by the matrix
$\frac{1}{2}
\left[\begin{smallmatrix} 1+i & 1-i \\ 1-i & 1+i \end{smallmatrix}\right]$.
We note that this is
 the composition of a $Z$-transformation and an $H$-transformation where $Z = \frac{1}{\sqrt{2}}\left[\begin{smallmatrix} 1 & 1 \\ i & -i \end{smallmatrix}\right]$ and $H = \frac{1}{\sqrt{2}}\left[\begin{smallmatrix} 1 & 1 \\ 1 & -1 \end{smallmatrix}\right]$, namely
$\frac{1}{2}
\left[\begin{smallmatrix} 1+i & 1-i \\ 1-i & 1+i \end{smallmatrix}\right]
= HZ$. Then 
%\begin{align*}
%Z_{\textup{8V}}(G; a, b, c, d) & = \textup{Holant}\left(G'; \neq_2 |\ f\right)\\
%& = \textup{Holant}\left(G'; \neq_2 \cdot \left( (HZ)^{-1} \right)^{\otimes 2} |\ (HZ)^{\otimes 4} \cdot f\right) \\
%& = \textup{Holant}\left(G'; =_2 |\ (HZ)^{\otimes 4} f\right)\\
%& = \textup{Holant}\left(G; \ (HZ)^{\otimes 4} f\right),
%\end{align*}
%and a direct calculation shows that $M((HZ)^{\otimes 4} f) = \frac{1}{2}
%\left[\begin{smallmatrix} a + b + c - d & 0 & 0 & - a + b + c + d \\ 0 & a - b + c + d & a + b - c + d & 0 \\ 0 & a + b - c + d & a - b + c + d & 0 \\ - a + b + c + d & 0 & 0 & a + b + c - d \end{smallmatrix}\right]$.
\begin{align*}
Z_{\textup{8V}}(G; a, b, c, d) & = \textup{Holant}\left(G'; \neq_2 |\ f\right)\\
& = \textup{Holant}\left(G'; (\neq_2) \cdot \left( (HZ)^{-1} \right)^{\otimes 2} |\ (HZ)^{\otimes 4} \cdot f\right) \\
& = \textup{Holant}\left(G'; =_2 |\ (HZ)^{\otimes 4} f\right)\\
& = \textup{Holant}\left(G; \ (HZ)^{\otimes 4} f\right).
\end{align*}
Here $(\neq_2) \cdot \left( (HZ)^{-1} \right)^{\otimes 2} 
= (\neq_2) \cdot \left( Z^{-1} \right)^{\otimes 2}  \cdot \left( H^{-1} \right)^{\otimes 2} 
= (=_2) \cdot \left( H^{-1} \right)^{\otimes 2}
= (=_2)$, because $H$ is orthogonal.
Now a direct calculation shows that $M((HZ)^{\otimes 4} f) = \frac{1}{2}
\left[\begin{smallmatrix} a + b + c - d & 0 & 0 & - a + b + c + d \\ 0 & a - b + c + d & a + b - c + d & 0 \\ 0 & a + b - c + d & a - b + c + d & 0 \\ - a + b + c + d & 0 & 0 & a + b + c - d \end{smallmatrix}\right]$.
\end{proof}

\bigskip

\section{Planar graphs}\label{sec:planar}
\begin{lemma}\label{lem:planar}
Let $G$ be a 4-regular plane graph. Then 
\(Z_{\textup{8V}}(G; a, b, c, d) = Z_{\textup{EC}}(G; b, a, d, c)\).
\end{lemma}
\begin{proof}
It is well known that a connected planar graph is Eulerian if and only if its dual is bipartite~\cite{WELSH1969375}. 
%As $G$ is an embedded 4-regular planar graph, $G$ is Eulerian.
Partition functions are multiplicative over connected components,
so we may assume $G$ is connected.
As $G$ is planar and 4-regular, the dual $G^*$ is bipartite.
 %Thus the dual $G^*$ is bipartite. 
Hence, we can color the faces of $G$ using two colors, say black and white, so that any two adjacent faces
(i.e., they share an edge) are of different colors. For definiteness, we assume that the outer face of $G$ is colored white.
See \figref{fig:planar_blackwhite} for an example.
Every edge separates one face colored white and another face colored black, so each edge is on a unique white face.
This shows that the binary relation on the set of edges
 defined by being on the same white face is an equivalence relation.
%This defines an equivalence relation of the edges such that two edges are 
%equivalent if and only if they are on the same white face.

\captionsetup[subfigure]{labelformat=parens}
\renewcommand{\thesubfigure}{\textsc{\alph{subfigure}}}
\begin{figure}[h!]
\centering
\begin{subfigure}[t]{0.42\linewidth}
\centering
\includegraphics[width=0.7\linewidth]{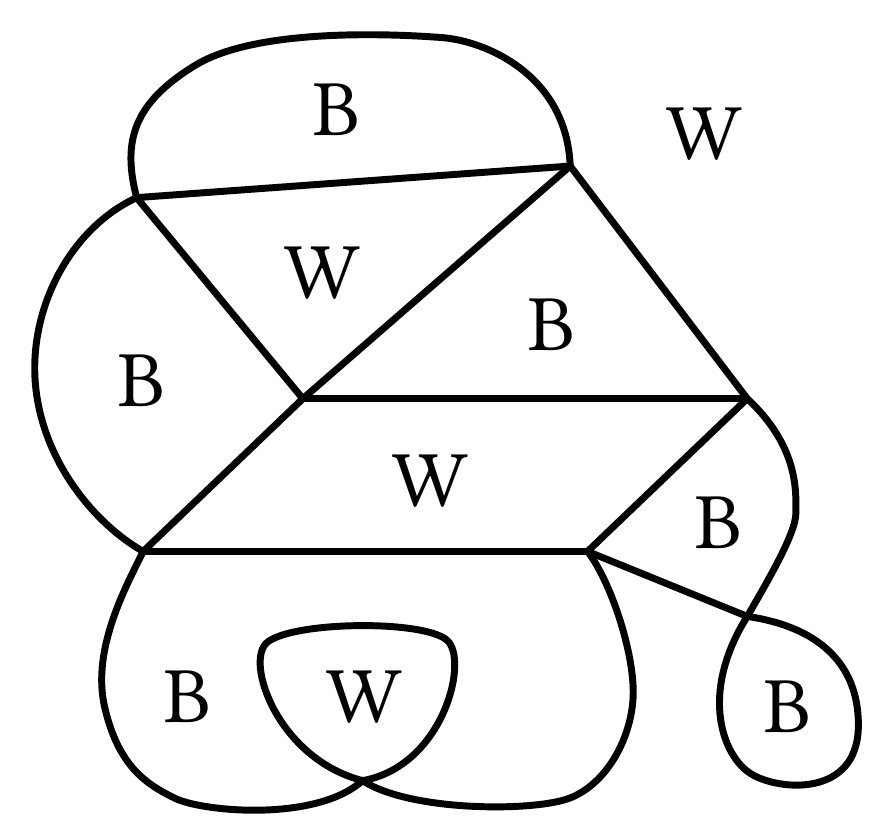}
%\caption{A proper 2-coloring of the faces of a planar 4-regular graph.}
\caption{}
\label{fig:planar_blackwhite}
\end{subfigure}
%\hspace{0.05\linewidth}
\begin{subfigure}[t]{0.42\linewidth}
\centering
\includegraphics[width=0.7\linewidth]{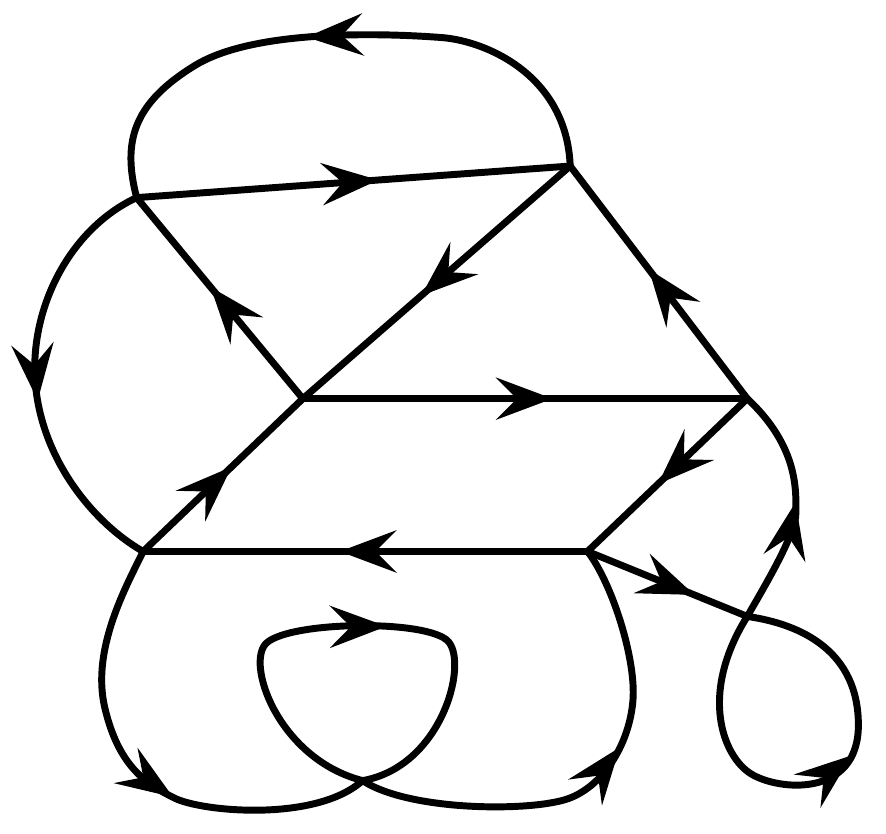}
%\caption{The canonical orientation.}
\caption{}
\label{fig:planar_canonical}
\end{subfigure}
\caption{A proper 2-coloring of the faces of a planar 4-regular graph and its canonical orientation $\tau$.}
\end{figure}

Based on the above facts, there is a \emph{canonical orientation} $\tau$ 
which orients all the edges along any (non-outer) white face \emph{clockwise}.
%This defines a unique orientation since every edge is on exactly one white face. 
This is also the same as to orient edges along every black face \emph{counterclockwise}. See \figref{fig:planar_canonical} for a pictorial illustration.
Observe that $\tau$ is an Eulerian orientation of $G$ (at every vertex the in-degree equals the out-degree) and in $\tau$ every vertex is in the 5th local configuration in \figref{fig:blackwhite} (and equivalently 
the 6th local configuration  in \figref{fig:whiteblack}).

\captionsetup[subfigure]{labelformat=parens}
\renewcommand{\thesubfigure}{\textsc{\alph{subfigure}}}
\begin{figure}[h!]
\centering
\begin{subfigure}[t]{0.42\linewidth}
\centering
\includegraphics[width=0.7\linewidth]{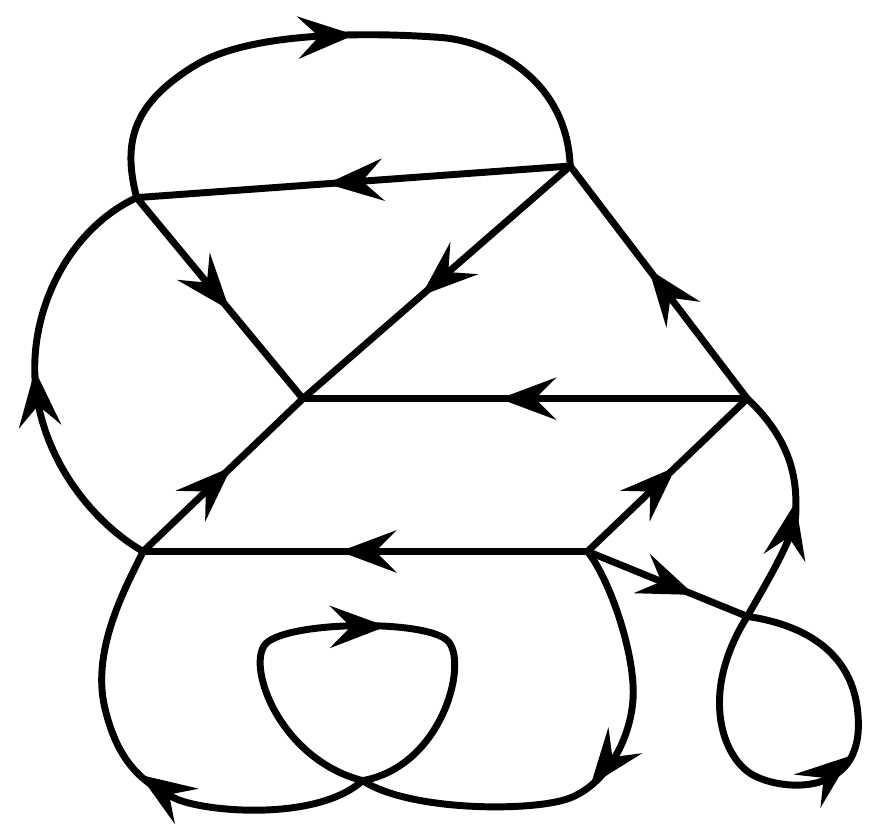}
\caption{}
%\caption{An even orientation.}
\label{fig:planar_generic}
\end{subfigure}
%\hspace{0.05\linewidth}
\begin{subfigure}[t]{0.42\linewidth}
\centering
\includegraphics[width=0.7\linewidth]{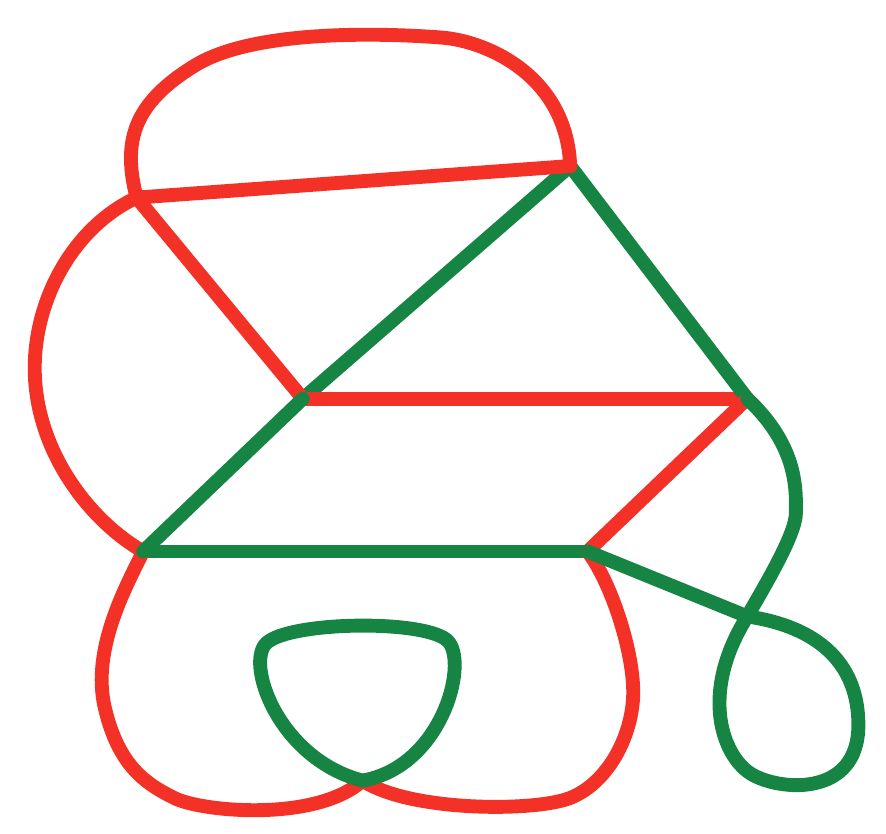}
\caption{}
%\caption{The canonical orientation.}
\label{fig:planar_coloring}
\end{subfigure}
\caption{An even orientation $\tau'$ and its corresponding even coloring.}
\label{fig:planar_coloring-figurefull}
\end{figure}

For an arbitrary orientation $\tau'$ of $G$, we say an edge $e$ is \emph{green} if $\tau'(e) = \tau(e)$, and $e$ is \emph{red} otherwise.
Then for any  even orientation $\tau'$ of $G$,
we can show that this assignment of green-red colors 
is an \emph{even coloring} of $G$.
(See an illustration in \figref{fig:planar_coloring-figurefull}: \figref{fig:planar_generic} is an even orientation and \figref{fig:planar_coloring} is its corresponding even coloring.)
Every even orientation $\tau'$  gives an even coloring because:
\begin{itemize}
\item
Under this coloring, the canonical orientation $\tau$ receives the all-green coloring, which is an even coloring itself.
\item
For any even orientation $\tau'$, the local configuration of $\tau'$ at any vertex differs from the local configuration of $\tau$ at the same vertex on an \emph{even} number of edges. These edges receive the red color and 
the others receive green. Therefore, at each vertex the color assignment will be in one of the states shown in \figref{fig:even-colorings}.
\end{itemize}

\captionsetup[subfigure]{labelformat=empty}
\renewcommand{\thesubfigure}{-\arabic{subfigure}}
\begin{figure}[h!]
\centering
\begin{subfigure}[b]{0.12\linewidth}
\centering\includegraphics[width=\linewidth]{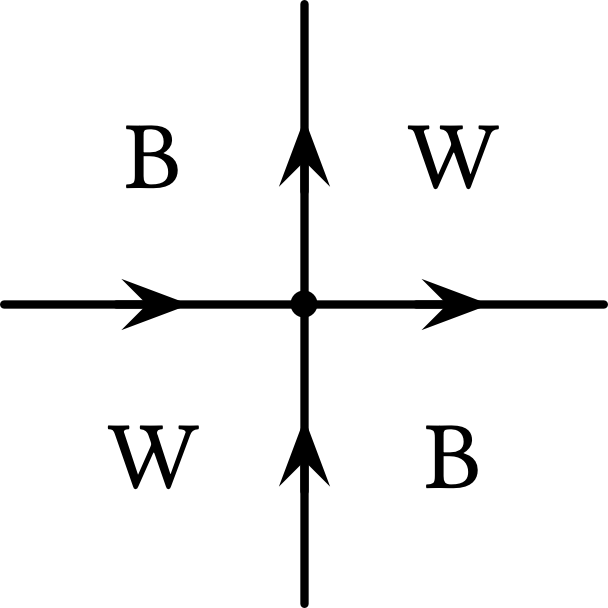}\caption{$1$}
\label{fig:blackwhite_1}
\end{subfigure}
\begin{subfigure}[b]{0.12\linewidth}
\centering\includegraphics[width=\linewidth]{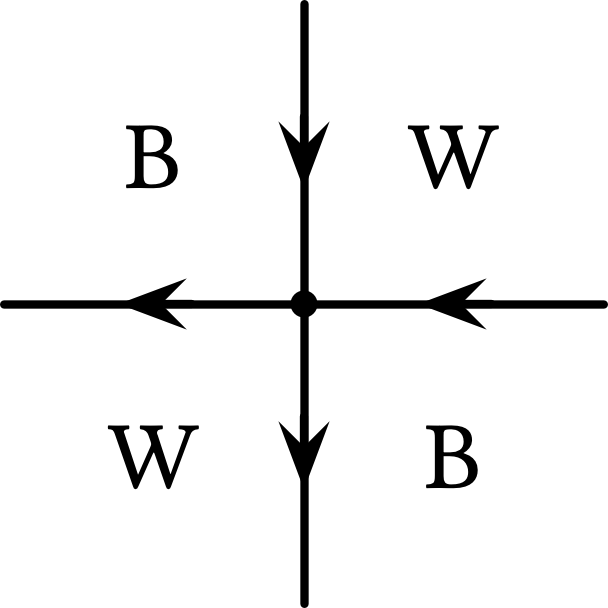}\caption{$2$}
\label{fig:blackwhite_2}
\end{subfigure}
\begin{subfigure}[b]{0.12\linewidth}
\centering\includegraphics[width=\linewidth]{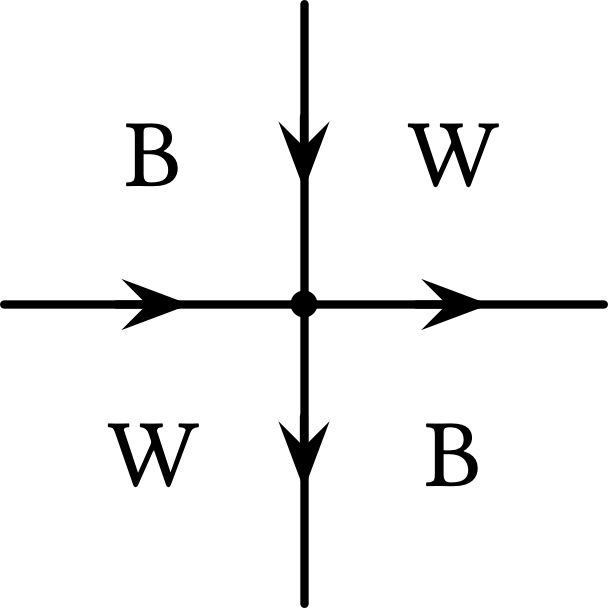}\caption{$3$}
\label{fig:blackwhite_3}
\end{subfigure}
\begin{subfigure}[b]{0.12\linewidth}
\centering\includegraphics[width=\linewidth]{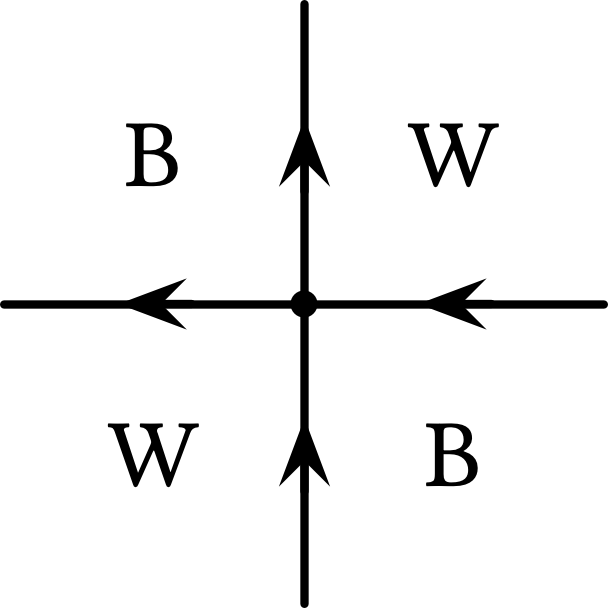}\caption{$4$}
\label{fig:blackwhite_4}
\end{subfigure}
\begin{subfigure}[b]{0.12\linewidth}
\centering\includegraphics[width=\linewidth]{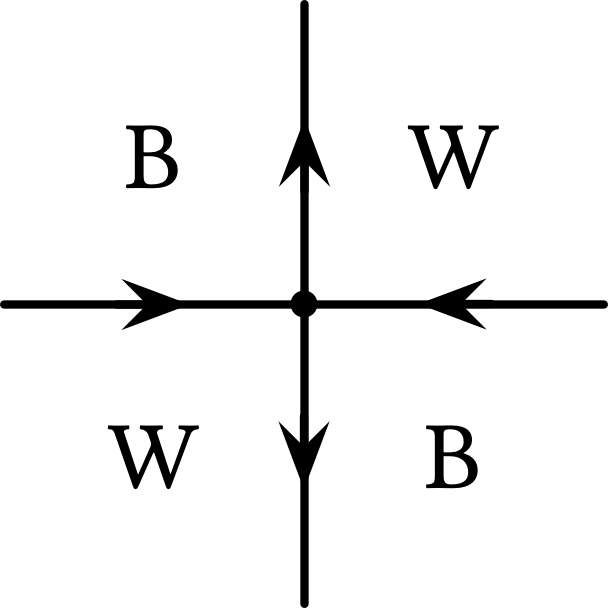}\caption{$5$}
\label{fig:blackwhite_5}
\end{subfigure}
\begin{subfigure}[b]{0.12\linewidth}
\centering\includegraphics[width=\linewidth]{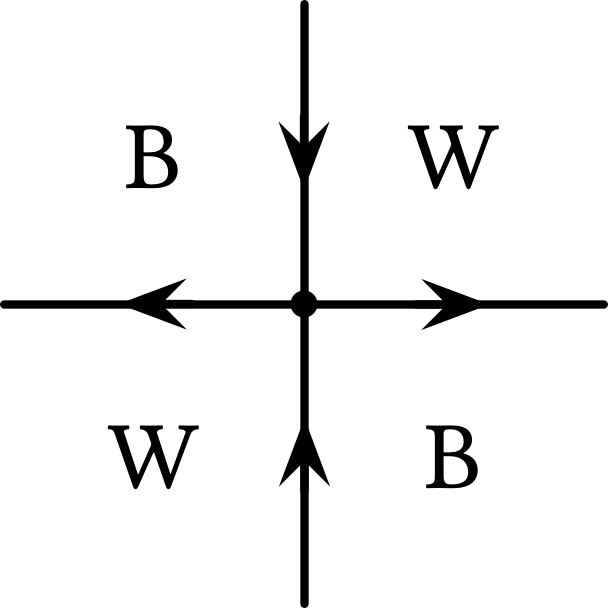}\caption{$6$}
\label{fig:blackwhite_6}
\end{subfigure}
\begin{subfigure}[b]{0.12\linewidth}
\centering\includegraphics[width=\linewidth]{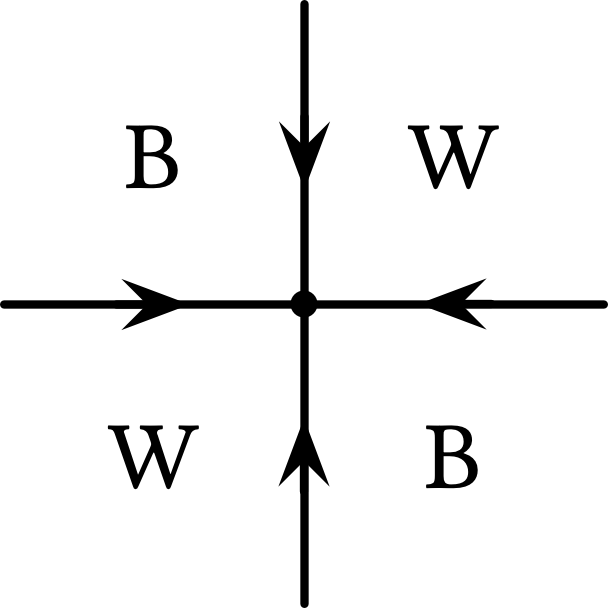}\caption{$7$}
\label{fig:blackwhite_7}
\end{subfigure}
\begin{subfigure}[b]{0.12\linewidth}
\centering\includegraphics[width=\linewidth]{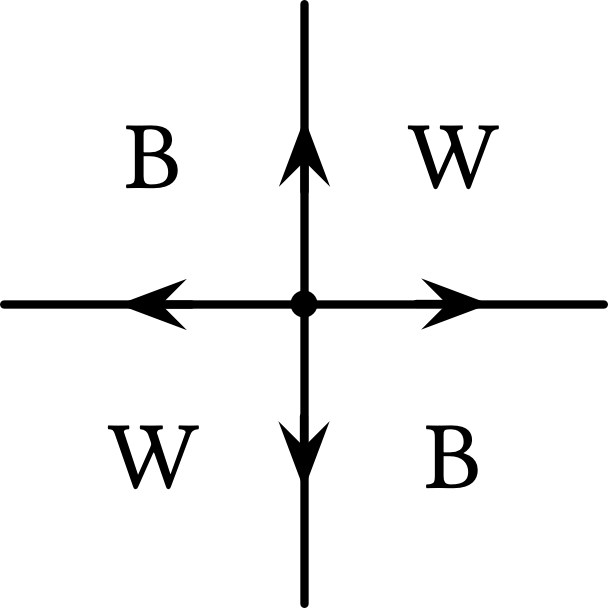}\caption{$8$}
\label{fig:blackwhite_8}
\end{subfigure}
\caption{}\label{fig:blackwhite}
\end{figure}

\begin{figure}[h!]
\centering
\begin{subfigure}[b]{0.12\linewidth}
\centering\includegraphics[width=\linewidth]{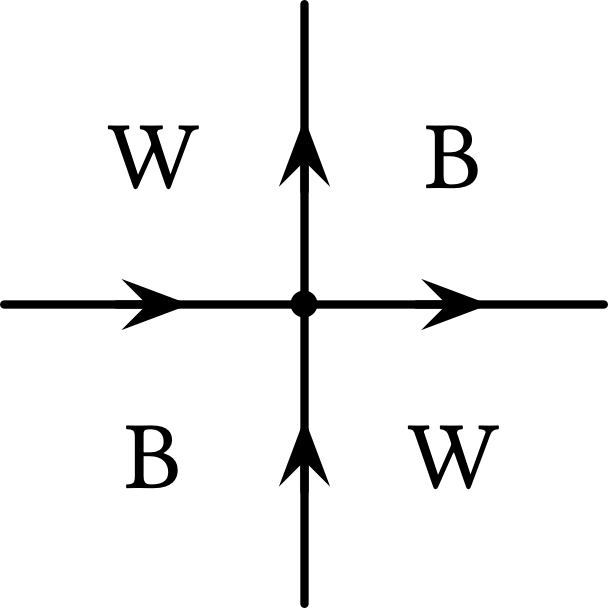}\caption{$1$}
\label{fig:whiteblack_1}
\end{subfigure}
\begin{subfigure}[b]{0.12\linewidth}
\centering\includegraphics[width=\linewidth]{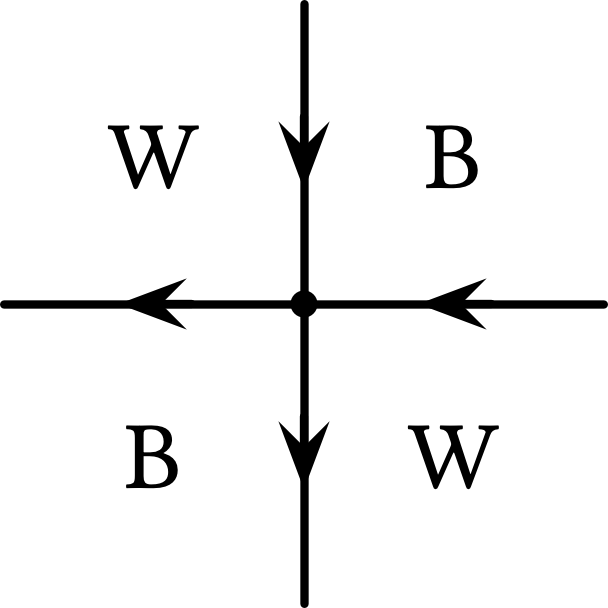}\caption{$2$}
\label{fig:whiteblack_2}
\end{subfigure}
\begin{subfigure}[b]{0.12\linewidth}
\centering\includegraphics[width=\linewidth]{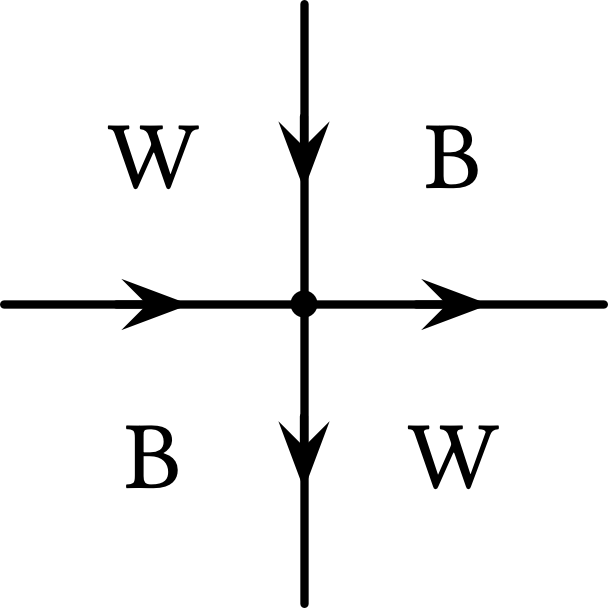}\caption{$3$}
\label{fig:whiteblack_3}
\end{subfigure}
\begin{subfigure}[b]{0.12\linewidth}
\centering\includegraphics[width=\linewidth]{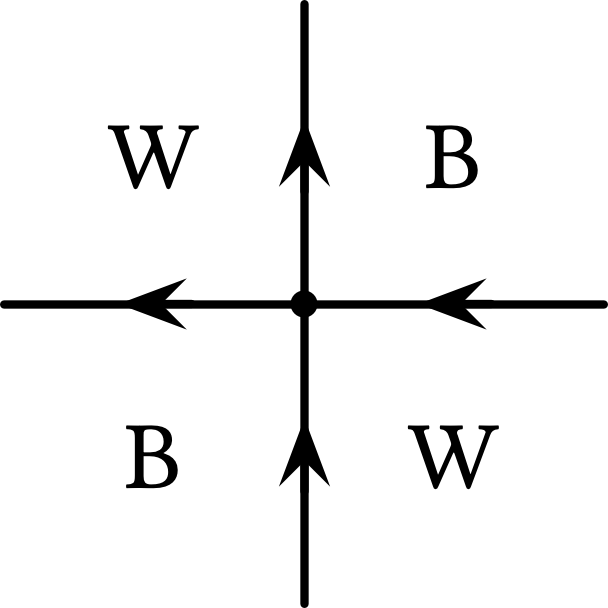}\caption{$4$}
\label{fig:whiteblack_4}
\end{subfigure}
\begin{subfigure}[b]{0.12\linewidth}
\centering\includegraphics[width=\linewidth]{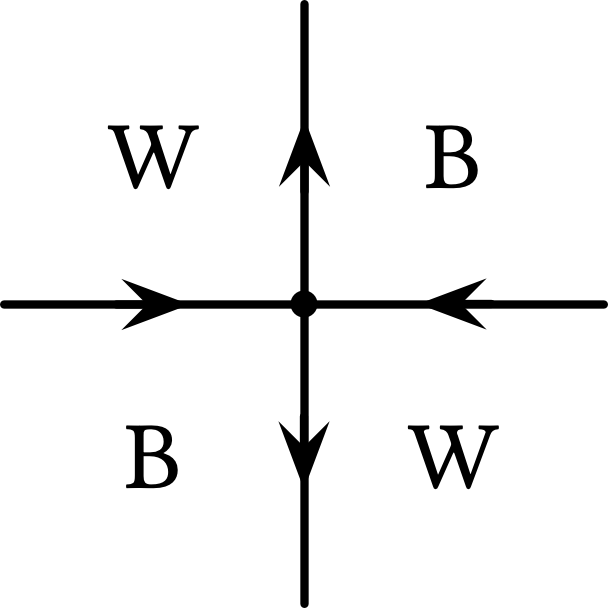}\caption{$5$}
\label{fig:whiteblack_5}
\end{subfigure}
\begin{subfigure}[b]{0.12\linewidth}
\centering\includegraphics[width=\linewidth]{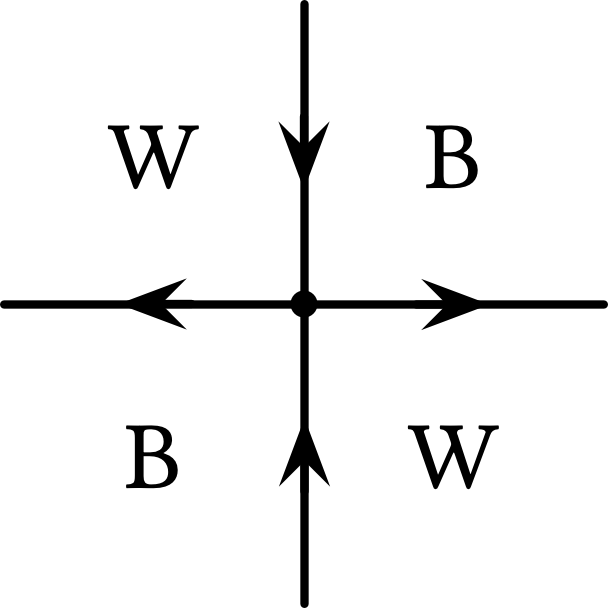}\caption{$6$}
\label{fig:whiteblack_6}
\end{subfigure}
\begin{subfigure}[b]{0.12\linewidth}
\centering\includegraphics[width=\linewidth]{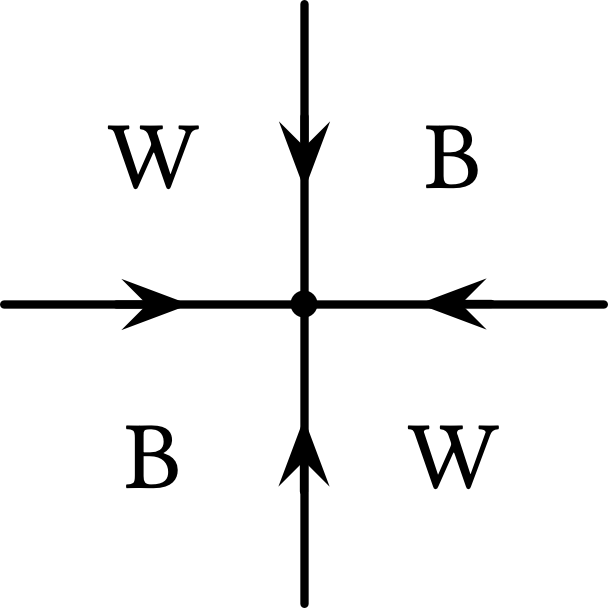}\caption{$7$}
\label{fig:whiteblack_7}
\end{subfigure}
\begin{subfigure}[b]{0.12\linewidth}
\centering\includegraphics[width=\linewidth]{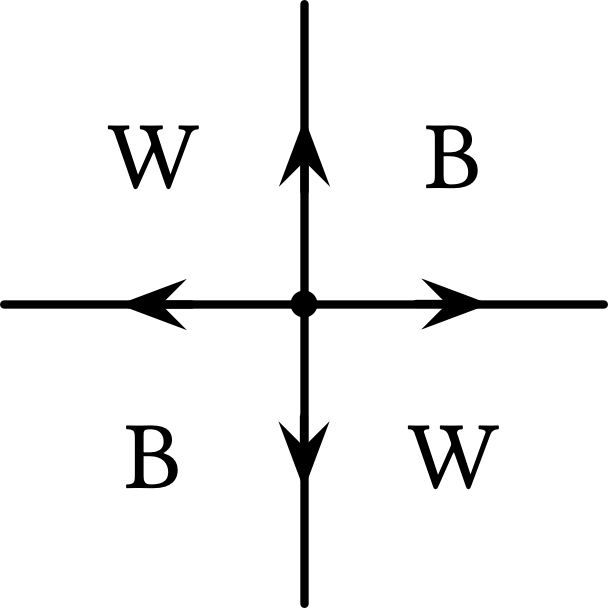}\caption{$8$}
\label{fig:whiteblack_8}
\end{subfigure}
\caption{}\label{fig:whiteblack}
\end{figure}

We claim that
this color assignment on the edges of $G$ gives a 
%one-to-one onto mapping
bijection $M_\textup{\textsc{Planar}}$ from $S_\textsc{8V}$ to $S_\textsc{EC}$.
%the set of (eight) valid local configurations around a vertex in the eight-vertex model and the set of (eight) valid local configurations around a vertex in the even-coloring model.
%This can be seen in the entry-wise correspondence
Given the black and white coloring of the faces of $G$,
there are two types of vertices in $G$, either the one in \figref{fig:blackwhite} or the other one in \figref{fig:whiteblack}.
The correspondence of local configurations from \figref{fig:blackwhite} to \figref{fig:even-colorings} is $(1, 2, 3, 4, 5, 6, 7, 8) \rightarrow (3, 4, 1, 2, 7, 8, 5, 6)$; the correspondence of local configurations from \figref{fig:whiteblack} to \figref{fig:even-colorings} is $(1, 2, 3, 4, 5, 6, 7, 8) \rightarrow (4, 3, 2, 1, 8, 7, 6, 5)$.
Consequently, $M_\textup{\textsc{Planar}}$ induces a one-to-one 
correspondence between $\mathcal{O}_{\bf e}(G)$ and $\mathcal{C}_{\bf e}(G)$.
%the set of even orientations and the set of valid configurations of the even-coloring model.

Although there are two types of maps, both of them map $\{1, 2\}$ to $\{3, 4\}$, $\{3, 4\}$ to $\{1, 2\}$, $\{5, 6\}$ to $\{7, 8\}$, and $\{7, 8\}$ to $\{5, 6\}$.
In terms of the weights of local configurations: the arrow reversal symmetry of the eight-vertex model induces an equivalence relation $\{1,2\},  \{3,4\}, \{5,6\}, \{7,8\}$ in \figref{fig:orientations}; the color reversal symmetry of the even-coloring model induces an equivalence relation $\{1,2\},  \{3,4\}, \{5,6\}, \{7,8\}$ in \figref{fig:even-colorings}. 
%It is the symmetric relations which allow $M_\textup{\textsc{Planar}}$ to be weight-preserving.
Thus there is a uniform and consistent way to assign a mapping
of parameters.
One can see that $M_\textup{\textsc{Planar}}$ is weight-preserving if the even-coloring model has parameter setting $(w, x, y, z) = (b, a, d, c)$.
 It follows that
\[Z_{\textup{8V}}(G; a, b, c, d) = Z_{\textup{EC}}(G; b, a, d, c).\]
\end{proof}

Now that we have set up multiple equations (\lemref{lem:pm-hard_holant}, \lemref{lem:holo_trans}, and \lemref{lem:planar}) between $Z_{\textup{8V}}$ and $Z_{\textup{EC}}$ under different mappings in terms of the parameter settings, we can combine them and obtain equations between $Z_{\textup{8V}}$ under different parameter settings.

Before that, we make the following observation.
Since in any even orientation of a 4-regular graph $G$, the number of sinks (\figref{fig:orientations_7}) must be equal to the number of sources (\figref{fig:orientations_8}) and thus their sum is always even, we know that in the eight-vertex model under parameter setting $(a, b, c, d)$, the weight of a state is unchanged if we flip $d$ to $-d$. Therefore, we have
\begin{equation}\label{eqn:d-flip}
Z_\textup{8V}(G; a, b, c, d) = Z_\textup{8V}(G; a, b, c, -d).
\end{equation}
(In particular, for non-negative $a, b, c, d$, even though $-d$ makes
an appearance on the right-hand-side, this equation says the value
$Z_\textup{8V}(G; a, b, c, -d) \ge 0$.)
Obviously the approximation complexity for computing $Z_\textup{8V}(a, b, c, -d)$ is the same as that for $Z_\textup{8V}(a, b, c, d)$.

\begin{notation}
Given a set of 4-tuples $S$, let $N_d(S) =  \{(a, b, c, -d) \ |\ (a, b, c, d) \in S\}$.
\end{notation}

\begin{notation}
Given two $4 \times 4$ invertible matrices $M_1$ and $M_2$, denote by $\left<M_1, M_2\right>$ the group of matrices generated by $M_1$ and $M_2$, where the group operation is matrix multiplication.
\end{notation}

\begin{theorem}\label{thm:planar}
Let $G$ be a 4-regular plane graph
and let $M_{Z}^\textup{\textsc{Pl}} = 
\frac{1}{2} \left[\begin{smallmatrix} 1 & -1 & 1 & -1 \\ -1 & 1 & 1 & -1 \\ 1 & 1 & 1 & 1 \\ 1 & 1 & -1 & -1 \end{smallmatrix}\right]$,
$M_{HZ}^\textup{\textsc{Pl}} = 
\frac{1}{2} \left[\begin{smallmatrix} 1 & -1 & 1 & 1 \\ -1 & 1 & 1 & 1 \\ 1 & 1 & 1 & -1 \\ 1 & 1 & -1 & 1 \end{smallmatrix}\right]$.
Then for any $M \in \left<M_{Z}^\textup{\textsc{Pl}}, M_{HZ}^\textup{\textsc{Pl}}\right>$,
$Z_\textup{8V}(G; a, b, c, d) = Z_\textup{8V}(G; a', b', c', d')$ where
$\left[\begin{smallmatrix} a'\\ b'\\ c'\\ d' \end{smallmatrix}\right]
= M
\left[\begin{smallmatrix} a\\ b\\ c\\ d \end{smallmatrix}\right]$.
\end{theorem}

\begin{remark}
One can check that $\left<M_{Z}^\textup{\textsc{Pl}}, M_{HZ}^\textup{\textsc{Pl}}\right>$ is isomorphic to the symmetry group $S_3$ and the group elements are shown in \tabref{tab:planar}. 
\end{remark}

\begin{proof}
First we prove the theorem for $M_{Z}^\textup{\textsc{Pl}}$. 
%that $Z_\textup{8V}(G; a, b, c, d) = Z_\textup{8V}(G; a', b', c', d')$ where
%$\left[\begin{smallmatrix} a'\\ b'\\ c'\\ d' \end{smallmatrix}\right]
%= M_{Z}^\textup{\textsc{Pl}}
%\left[\begin{smallmatrix} a\\ b\\ c\\ d \end{smallmatrix}\right]$. 
 From \lemref{lem:pm-hard_holant}, we know that
$Z_{\textup{8V}}(G; a, b, c, d) = Z_{\textup{EC}}(G; w, x, y, z)$ where
$\left[\begin{smallmatrix} w\\ x\\ y\\ z \end{smallmatrix}\right]
= \frac{1}{2} \left[\begin{smallmatrix} -1 & 1 & 1 & -1 \\ 1 & -1 & 1 & -1 \\ 1 & 1 & -1 & -1 \\ 1 & 1 & 1 & 1 \end{smallmatrix}\right]
\left[\begin{smallmatrix} a\\ b\\ c\\ d \end{smallmatrix}\right]$.
From \lemref{lem:planar}, we know that
$Z_{\textup{EC}}(G; w, x, y, z) = Z_{\textup{8V}}(G; a', b', c', d')$
where
$\left[\begin{smallmatrix} a'\\ b'\\ c'\\ d' \end{smallmatrix}\right] = \left[\begin{smallmatrix} 0 & 1 & 0 & 0  \\ 1 & 0 & 0 & 0 \\ 0 & 0 & 0 & 1 \\ 0 & 0 & 1 & 0 \end{smallmatrix}\right]
\left[\begin{smallmatrix} w\\ x\\ y\\ z \end{smallmatrix}\right]$.
Therefore, we have
$Z_\textup{8V}(G; a, b, c, d) = Z_\textup{8V}(G; a', b', c', d')$ where
$\left[\begin{smallmatrix} a'\\ b'\\ c'\\ d' \end{smallmatrix}\right]
= 
\left[\begin{smallmatrix} 0 & 1 & 0 & 0  \\ 1 & 0 & 0 & 0 \\ 0 & 0 & 0 & 1 \\ 0 & 0 & 1 & 0 \end{smallmatrix}\right] \cdot
\frac{1}{2} \left[\begin{smallmatrix} -1 & 1 & 1 & -1 \\ 1 & -1 & 1 & -1 \\ 1 & 1 & -1 & -1 \\ 1 & 1 & 1 & 1 \end{smallmatrix}\right]
\left[\begin{smallmatrix} a\\ b\\ c\\ d \end{smallmatrix}\right]
= M_{Z}^\textup{\textsc{Pl}}
\left[\begin{smallmatrix} a\\ b\\ c\\ d \end{smallmatrix}\right]$.

The proof for $M_{HZ}^\textup{\textsc{Pl}}$ is similar. Instead of combining \lemref{lem:pm-hard_holant} and \lemref{lem:planar}, we simply need to combine \lemref{lem:holo_trans} and \lemref{lem:planar} and notice that $M_{HZ}^\textup{\textsc{Pl}} = 
\left[\begin{smallmatrix} 0 & 1 & 0 & 0  \\ 1 & 0 & 0 & 0 \\ 0 & 0 & 0 & 1 \\ 0 & 0 & 1 & 0 \end{smallmatrix}\right] \cdot
\frac{1}{2} \left[\begin{smallmatrix} -1 & 1 & 1 & 1 \\ 1 & -1 & 1 & 1 \\ 1 & 1 & -1 & 1 \\ 1 & 1 & 1 & -1 \end{smallmatrix}\right]$.

Since the theorem is proved for the two invertible matrices $M_{Z}^\textup{\textsc{Pl}}$ and $M_{HZ}^\textup{\textsc{Pl}}$, it is also true for the group of matrices generated by these two matrices using their inverse and matrix multiplication.
\end{proof}

\begin{notation}
In order to state the results in this section and the next section, we adopt the following notations assuming $a, b, c, d \ge 0$.
\begin{itemize}
\item
$\mathcal{A} = \{(a,b,c,d) \; | \; a \le b+c+d\}$, $\mathcal{B} = \{(a,b,c,d) \; | \; b \le a+c+d\}$, $\mathcal{C} = \{(a,b,c,d) \; | \; c \le a+b+d\}$, $\mathcal{D} = \{(a,b,c,d) \; | \; d \le a+b+c\}$;
\item
$\mathcal{AD} = \{(a,b,c,d) \; | \; a+d \le b+c\}$, $\mathcal{BD} = \{(a,b,c,d) \; | \; b+d \le a+c\}$, $\mathcal{CD} = \{(a,b,c,d) \; | \; c+d \le a+b\}$.
\end{itemize}
\end{notation}
\begin{remark}
$\mathcal{AD} \subset \mathcal{A} \bigcap \mathcal{D}$, $\mathcal{BD} \subset \mathcal{B} \bigcap \mathcal{D}$, $\mathcal{CD} \subset \mathcal{C} \bigcap \mathcal{D}.$
$\mathcal{X} = \mathcal{A} \bigcap \mathcal{B} \bigcap \mathcal{C} \bigcap \mathcal{D}$.
$\mathcal{Y} = \mathcal{AD} \bigcap \mathcal{BD} \bigcap \mathcal{CD}$.
In addition, we abuse the notation and use 
$\overline{\mathcal{C}} = \{(a,b,c,d) \; | \; c \ge a+b+d\}$, $\overline{\mathcal{D}} = \{(a,b,c,d) \; | \; d \ge a+b+c\}$, $\overline{\mathcal{AD}} = \{(a,b,c,d) \; | \; a+d \ge b+c\}$, $\overline{\mathcal{BD}} = \{(a,b,c,d) \; | \; b+d \ge a+c\}$, and $\overline{\mathcal{CD}} = \{(a,b,c,d) \; | \; c+d \ge a+b\}$ in \tabref{tab:planar} and \tabref{tab:bipartite}.
\end{remark}

\begin{table}[htpb]
\centering
\caption{Elements of $\left<M_{Z}^\textup{\textsc{Pl}}, M_{HZ}^\textup{\textsc{Pl}}\right>$ and preimages of $\mathcal{Y} = \mathcal{AD} \bigcap \mathcal{BD} \bigcap \mathcal{CD}$ under corresponding maps.
A substantial subregion of each preimage admits FPRAS on planar 4-regular graphs.
The last column lists the approximation complexity of the eight-vertex model on general 4-regular graphs in the corresponding preimage regions.}
\label{tab:planar}
\ra{2.2}
\begin{tabular}{@{}l@{\phantom{abcde}}l@{\phantom{abcde}}l@{\phantom{abcde}}l@{}}
\toprule
Element & Matrix & Preimage & Approximation \\
\midrule
$I_4$
&
$\left[\begin{smallmatrix} 1 & 0 & 0 & 0 \\ 0 & 1 & 0 & 0 \\ 0 & 0 & 1 & 0 \\ 0 & 0 & 0 & 1 \end{smallmatrix}\right]$
&
$\mathcal{AD} \bigcap \mathcal{BD} \bigcap \mathcal{CD}$
&
FPRAS in $\mathcal{Z}$
\\
$M_{Z}^\textup{\textsc{Pl}}$
&
$\frac{1}{2} \left[\begin{smallmatrix} 1 & -1 & 1 & -1 \\ -1 & 1 & 1 & -1 \\ 1 & 1 & 1 & 1 \\ 1 & 1 & -1 & -1 \end{smallmatrix}\right]$
&
$N_d\left(\mathcal{AD} \bigcap \mathcal{BD} \bigcap \overline{\mathcal{CD}} \bigcap \mathcal{C}\right)$
&
\#PM-hard
\\
$\left(M_{Z}^\textup{\textsc{Pl}}\right)^2$
&
$\frac{1}{2} \left[\begin{smallmatrix} 1 & -1 & 1 & 1 \\ -1 & 1 & 1 & 1 \\ 1 & 1 & 1 & -1 \\ -1 & -1 & 1 & -1 \end{smallmatrix}\right]$
&
$\overline{\mathcal{C}}$
&
NP-hard
\\
$M_{HZ}^\textup{\textsc{Pl}}$ 
&
$\frac{1}{2} \left[\begin{smallmatrix} 1 & -1 & 1 & 1 \\ -1 & 1 & 1 & 1 \\ 1 & 1 & 1 & -1 \\ 1 & 1 & -1 & 1 \end{smallmatrix}\right]$
&
$\mathcal{AD} \bigcap \mathcal{BD} \bigcap \overline{\mathcal{CD}} \bigcap \mathcal{C}$
&
\#PM-hard
\\
$M_{Z}^\textup{\textsc{Pl}} M_{HZ}^\textup{\textsc{Pl}}$
&
$\frac{1}{2} \left[\begin{smallmatrix} 1 & -1 & 1 & -1 \\ -1 & 1 & 1 & -1 \\ 1 & 1 & 1 & 1 \\ -1 & -1 & 1 & 1 \end{smallmatrix}\right]$
&
$N_d\left(\overline{\mathcal{C}}\right)$
&
NP-hard
\\
$\left(M_{Z}^\textup{\textsc{Pl}}\right)^2 M_{HZ}^\textup{\textsc{Pl}}$
&
$\left[\begin{smallmatrix} 1 & 0 & 0 & 0 \\ 0 & 1 & 0 & 0 \\ 0 & 0 & 1 & 0 \\ 0 & 0 & 0 & -1 \end{smallmatrix}\right]$
&
$N_d\left(\mathcal{AD} \bigcap \mathcal{BD} \bigcap \mathcal{CD}\right)$
&
FPRAS in $N_d\left(\mathcal{Z}\right)$
\\
\bottomrule
\end{tabular}
\end{table}

\begin{corollary}\label{cor:planar}
Let $G$ be a 4-regular plane graph
and let $M_{Z}^\textup{\textsc{Pl}}$ and $M_{HZ}^\textup{\textsc{Pl}}$ be defined as in \thmref{thm:planar}.
Then for any $M \in \left<M_{Z}^\textup{\textsc{Pl}}, M_{HZ}^\textup{\textsc{Pl}}\right>$, there is an FPRAS for $Z_{\textup{8V}}(G; a, b, c, d)$ if $M \left[\begin{smallmatrix} a\\ b\\ c\\ d \end{smallmatrix}\right] \in \mathcal{Y} \bigcap \mathcal{Z}$.
\end{corollary}

Thus we know that for any $M \in \left<M_{Z}^\textup{\textsc{Pl}}, M_{HZ}^\textup{\textsc{Pl}}\right>$, there is an FPRAS for $Z_{\textup{8V}}(G; a, b, c, d)$ for $(a, b, c, d)$ in a subregion of $M^{-1}(\mathcal{Y})$.
Note that $\mathcal{AD} \bigcap \mathcal{BD} \bigcap \overline{\mathcal{CD}} \bigcap \mathcal{C} \subset \mathcal{X} \bigcap \overline{\mathcal{Y}}$ and $\overline{\mathcal{C}} \subset \overline{\mathcal{X}}$.
With the help of \tabref{tab:planar}, one can see that \thmref{thm:main_planar} is implied by \corref{cor:planar}\footnote{In fact, to prove \thmref{thm:main_planar} we need $\left(M_{HZ}^\textup{\textsc{Pl}}\right)^{-1} \left( \mathcal{Y} \bigcap \mathcal{Z} \right) \subset \mathcal{AD} \bigcap \mathcal{BD} \bigcap \overline{\mathcal{CD}} \bigcap \mathcal{C} \bigcap \overline{\mathcal{Z}}$ and one can check that this is true.}.

\bigskip

\section{Bipartite graphs}\label{sec:bipartite}
\begin{lemma}\label{lem:bipartite}
Let $G$ be a bipartite graph. Let $G'$ denote its edge-vertex incidence graph.
Suppose $f$ satisfies arrow reversal symmetry.
Then $\textup{Holant}\left(G'; \neq_2 |\ f\right) = \textup{Holant}\left(G'; =_2 |\ f\right) = \textup{Holant}\left(G; f\right)$.
In particular, if $G$ is 4-regular, then
%$Z_{\textup{8V}}(G; a, b, c, d) = \textup{Holant}\left(G; \left[\begin{smallmatrix} d & 0 & 0 & a \\ 0 & b & c & 0 \\ 0 & c & b & 0 \\ a & 0 & 0 & d \end{smallmatrix}\right]\right)$.
$Z_{\textup{8V}}(G; a, b, c, d) = Z_{\textup{EC}}(G; a, b, c, d)$.
\end{lemma}
\begin{proof}
For any bipartite graph $G = (L, R, E)$,
%For the case where $G = (L, R, E)$ is a 4-regular bipartite graph, 
%our proof is similar to the proof of \lemref{lem:planar}.
%We note that on bipartite graphs, 
there is a \emph{canonical orientation} $\tau$ which is to orient all the edges from $R$ to $L$.
Let $G$ be a 4-regular bipartite graph.
Observe that in $\tau$ every vertex in $L$ has local configuration \figref{fig:orientations_7} and every vertex in $R$ has local configuration \figref{fig:orientations_8}.

For an arbitrary orientation $\tau'$ of $G$, we say an edge $e$ is \emph{green} if $\tau'(e) = \tau(e)$, and $e$ is \emph{red} otherwise.
This coloring assignment on the edges of $G$ gives
 a 
%one-to-one onto mapping
bijection $M_\textup{\textsc{Bipartite}}$ from $S_\textsc{8V}$ to $S_\textsc{EC}$.
For the vertices in $L$, this can be seen in the entry-wise correspondence
$(1, 2, 3, 4, 5, 6, 7, 8) \rightarrow (1, 2, 3, 4, 5, 6, 7, 8)$ 
 from \figref{fig:orientations} to \figref{fig:even-colorings}; for the vertices in $R$, the correspondence of local configurations from \figref{fig:orientations} to \figref{fig:even-colorings} is $(1, 2, 3, 4, 5, 6, 7, 8) \rightarrow (2, 1, 4, 3, 6, 5, 8, 7)$.
Consequently
$M_\textup{\textsc{Bipartite}}$ defines a one-to-one correspondence
 between $\mathcal{O}_{\bf e}(G)$ and $\mathcal{C}_{\bf e}(G)$.
Again because both maps respect the same equivalence relation
induced by the arrow reversal symmetry,
this one-to-one correspondence $M_\textup{\textsc{Bipartite}}$ is weight-preserving if the even-coloring model has the parameter setting $(w, x, y, z) = (a, b, c, d)$. It follows that
\[Z_{\textup{8V}}(G; a, b, c, d) = Z_{\textup{EC}}(G; a, b, c, d).\]

The idea of the above mapping can be easily extended to general (not necessarily 4-regular) graphs.
For a bipartite graph $G = (L, R, E)$ and its edge-vertex incidence graph $G' = (V_E, L \cup R, E')$, every vertex $v_e \in V_E$ has degree 2 and connects a vertex $l \in L$ with a vertex $r \in R$.
To see a one-to-one weight-preserving mapping from valid configurations in $\textup{Holant}\left(G'; \neq_2 |\ f\right)$ to valid configurations in $\textup{Holant}\left(G'; =_2 |\ f\right)$, one simply flips the assignment on every edge $\{v_e, r\}$ such that $v_e \in V_E$ and $r \in R$.
\end{proof}

For any bipartite regular graph $G = (L, R, E)$, we know that $|L| = |R|$ and hence the total number of vertices is always even.
In the eight-vertex model under parameter setting $(a, b, c, d)$, the weight of a state is unchanged if we flip the sign of the weight on every vertex.
Therefore, for bipartite 4-regular graphs, in addition to (\ref{eqn:d-flip}), we also have
\begin{equation}\label{eqn:all-flip}
Z_\textup{8V}(G; a, b, c, d) = Z_\textup{8V}(G; -a, -b, -c, -d).
\end{equation}

\begin{notation}
Given a set of 4-tuples $S$, let $N(S) =  \{(-a, -b, -c, -d) \ |\ (a, b, c, d) \in S\}$.
\end{notation}

\begin{theorem}\label{thm:bipartite}
Let $G$ be a 4-regular bipartite graph
and let $M_{Z}^\textup{\textsc{Bi}} = 
\frac{1}{2} \left[\begin{smallmatrix} -1 & 1 & 1 & -1 \\ 1 & -1 & 1 & -1 \\ 1 & 1 & -1 & -1 \\ 1 & 1 & 1 & 1 \end{smallmatrix}\right]$,
$M_{HZ}^\textup{\textsc{Bi}} = 
\frac{1}{2} \left[\begin{smallmatrix} -1 & 1 & 1 & 1 \\ 1 & -1 & 1 & 1 \\ 1 & 1 & -1 & 1 \\ 1 & 1 & 1 & -1 \end{smallmatrix}\right]$.
Then for any $M \in \left<M_{Z}^\textup{\textsc{Bi}}, M_{HZ}^\textup{\textsc{Bi}}\right>$,
$Z_{\textup{8V}}(G; a, b, c, d) = Z_{\textup{8V}}(G; a', b', c', d')$ where
$\left[\begin{smallmatrix} a'\\ b'\\ c'\\ d' \end{smallmatrix}\right]
= M
\left[\begin{smallmatrix} a\\ b\\ c\\ d \end{smallmatrix}\right]$.
\end{theorem}

\begin{remark}
One can check that $\left<M_{Z}^\textup{\textsc{Bi}}, M_{HZ}^\textup{\textsc{Bi
}}\right>$ is isomorphic to the dihedral group $D_6$ and the group elements are shown in \tabref{tab:bipartite}. 
\end{remark}

\begin{proof}
The proof is similar to that of \thmref{thm:planar}.
Instead of combining the holographic maps in \lemref{lem:pm-hard_holant}, \lemref{lem:holo_trans} with the planar map in \lemref{lem:planar}, we need to combine them with the bipartite map in \lemref{lem:bipartite}.
\end{proof}

\begin{table}[htpb]
\centering
\caption{Elements of $\left<M_{Z}^\textup{\textsc{Bi}}, M_{HZ}^\textup{\textsc{Bi}}\right>$ and preimages of $\mathcal{Y} = \mathcal{AD} \bigcap \mathcal{BD} \bigcap \mathcal{CD}$ under corresponding maps.
A substantial subregion of each preimage admits FPRAS on bipartite 4-regular graphs.
The last column lists the approximation complexity of the eight-vertex model on general 4-regular graphs in the corresponding preimage regions.}
\label{tab:bipartite}
\ra{2.2}
\begin{tabular}{@{}l@{\phantom{abcde}}l@{\phantom{abcde}}l@{\phantom{abcde}}l@{}}
\toprule
Element & Matrix & Preimage & Approximation \\
\midrule
$I_4$
&
$\left[\begin{smallmatrix} 1 & 0 & 0 & 0 \\ 0 & 1 & 0 & 0 \\ 0 & 0 & 1 & 0 \\ 0 & 0 & 0 & 1 \end{smallmatrix}\right]$
&
$\mathcal{AD} \bigcap \mathcal{BD} \bigcap \mathcal{CD}$
&
FPRAS in $\mathcal{Z}$
\\
$M_{Z}^\textup{\textsc{Bi}}$
&
$\frac{1}{2} \left[\begin{smallmatrix} -1 & 1 & 1 & -1 \\ 1 & -1 & 1 & -1 \\ 1 & 1 & -1 & -1 \\ 1 & 1 & 1 & 1 \end{smallmatrix}\right]$
&
$N_d\left(\overline{\mathcal{AD}} \bigcap \overline{\mathcal{BD}} \bigcap \overline{\mathcal{CD}} \bigcap \mathcal{D}\right)$
&
\#PM-hard
\\
$\left(M_{Z}^\textup{\textsc{Bi}}\right)^2$
&
$\frac{1}{2} \left[\begin{smallmatrix} 1 & -1 & -1 & -1 \\ -1 & 1 & -1 & -1 \\ -1 & -1 & 1 & -1 \\ 1 & 1 & 1 & -1 \end{smallmatrix}\right]$
&
$N\left(\overline{\mathcal{D}}\right)$
&
NP-hard
\\
$\left(M_{Z}^\textup{\textsc{Bi}}\right)^3$
&
$-I_4$
&
$N\left(\mathcal{AD} \bigcap \mathcal{BD} \bigcap \mathcal{CD}\right)$
&
FPRAS in $N\left(\mathcal{Z}\right)$
\\
$\left(M_{Z}^\textup{\textsc{Bi}}\right)^4$
&
$-M_{Z}^\textup{\textsc{Bi}}$
&
$N\left(N_d\left(\overline{\mathcal{AD}} \bigcap \overline{\mathcal{BD}} \bigcap \overline{\mathcal{CD}} \bigcap \mathcal{D}\right)\right)$
&
\#PM-hard
\\
$\left(M_{Z}^\textup{\textsc{Bi}}\right)^5$
&
$-\left(M_{Z}^\textup{\textsc{Bi}}\right)^2$
&
$\overline{\mathcal{D}}$
&
NP-hard
\\
$M_{HZ}^\textup{\textsc{Bi}}$ 
&
$\frac{1}{2} \left[\begin{smallmatrix} -1 & 1 & 1 & 1 \\ 1 & -1 & 1 & 1 \\ 1 & 1 & -1 & 1 \\ 1 & 1 & 1 & -1 \end{smallmatrix}\right]$
&
$\overline{\mathcal{AD}} \bigcap \overline{\mathcal{BD}} \bigcap \overline{\mathcal{CD}} \bigcap \mathcal{D}$
&
\#PM-hard
\\
$M_{Z}^\textup{\textsc{Bi}} M_{HZ}^\textup{\textsc{Bi}}$
&
$\frac{1}{2} \left[\begin{smallmatrix} 1 & -1 & -1 & 1 \\ -1 & 1 & -1 & 1 \\ -1 & -1 & 1 & 1 \\ 1 & 1 & 1 & 1 \end{smallmatrix}\right]$
&
$N\left(N_d\left(\overline{\mathcal{D}}\right)\right)$
&
NP-hard
\\
$\left(M_{Z}^\textup{\textsc{Bi}}\right)^2 M_{HZ}^\textup{\textsc{Bi}}$
&
$\left[\begin{smallmatrix} -1 & 0 & 0 & 0 \\ 0 & -1 & 0 & 0 \\ 0 & 0 & -1 & 0 \\ 0 & 0 & 0 & 1 \end{smallmatrix}\right]$
&
$N\left(N_d\left(\mathcal{AD} \bigcap \mathcal{BD} \bigcap \mathcal{CD}\right)\right)$
&
FPRAS in $N\left(N_d\left(\mathcal{Z}\right)\right)$
\\
$\left(M_{Z}^\textup{\textsc{Bi}}\right)^3 M_{HZ}^\textup{\textsc{Bi}}$ 
&
$-M_{HZ}^\textup{\textsc{Bi}}$
&
$N\left(\overline{\mathcal{AD}} \bigcap \overline{\mathcal{BD}} \bigcap \overline{\mathcal{CD}} \bigcap \mathcal{D}\right)$
&
\#PM-hard
\\
$\left(M_{Z}^\textup{\textsc{Bi}}\right)^4 M_{HZ}^\textup{\textsc{Bi}}$
&
$-M_{Z}^\textup{\textsc{Bi}} M_{HZ}^\textup{\textsc{Bi}}$
&
$N_d\left(\overline{\mathcal{D}}\right)$
&
NP-hard
\\
$\left(M_{Z}^\textup{\textsc{Bi}}\right)^5 M_{HZ}^\textup{\textsc{Bi}}$
&
$-\left(M_{Z}^\textup{\textsc{Bi}}\right)^2 M_{HZ}^\textup{\textsc{Bi}}$
&
$N_d\left(\mathcal{AD} \bigcap \mathcal{BD} \bigcap \mathcal{CD}\right)$
&
FPRAS in $N_d\left(\mathcal{Z}\right)$
\\
\bottomrule
\end{tabular}
\end{table}

\begin{corollary}\label{cor:bipartite}
Let $G$ be a 4-regular bipartite graph
and let $M_{Z}^\textup{\textsc{Bi}}$ and $M_{HZ}^\textup{\textsc{Bi}}$ be defined as in \thmref{thm:bipartite}.
Then for any $M \in \left<M_{Z}^\textup{\textsc{Bi}}, M_{HZ}^\textup{\textsc{Bi}}\right>$, there is an FPRAS for $Z_{\textup{8V}}(G; a, b, c, d)$ if $M \left[\begin{smallmatrix} a\\ b\\ c\\ d \end{smallmatrix}\right] \in \mathcal{Y} \bigcap \mathcal{Z}$.
\end{corollary}

Thus we know that for any $M \in \left<M_{Z}^\textup{\textsc{Bi}}, M_{HZ}^\textup{\textsc{Bi}}\right>$, there is an FPRAS for $Z_{\textup{8V}}(G; a, b, c, d)$ for $(a, b, c, d)$ in a subregion of $M^{-1}(\mathcal{Y})$.
Note that $\overline{\mathcal{AD}} \bigcap \overline{\mathcal{BD}} \bigcap \overline{\mathcal{CD}} \bigcap \mathcal{D} \subset \mathcal{X} \bigcap \overline{\mathcal{Y}}$ and $\overline{\mathcal{D}} \subset \overline{\mathcal{X}}$.
With the help of \tabref{tab:bipartite}, one can see that \thmref{thm:main_bipartite} is implied by \corref{cor:bipartite}.

\bigskip

\section{Concluding remarks}{
All the FPRAS results obtained in this paper come from the algorithm for $\mathcal{AD} \bigcap \mathcal{BD} \bigcap \mathcal{CD} \bigcap \mathcal{Z}$. It is open if there exists an FPRAS/FPTAS for all $(a, b, c, d) \in \mathcal{AD} \bigcap \mathcal{BD} \bigcap \mathcal{CD}$. Assuming such an algorithm exists, our maps in \secref{sec:planar} would imply that all $\mathcal{AD} \bigcap \mathcal{BD}$ is approximable on planar graphs, and our maps in \secref{sec:bipartite} would imply that all $\overline{\mathcal{AD}} \bigcap \overline{\mathcal{BD}} \bigcap \overline{\mathcal{CD}}$ and $\overline{\mathcal{D}}$ are approximable on bipartite graphs.
We  note that the approximation in $\overline{\mathcal{A}}$, $\overline{\mathcal{B}}$, and $\overline{\mathcal{C}}$ is proved to be NP-hard even on bipartite graphs~\cite{DBLP:journals/corr/abs-1811-03126}.

In \secref{sec:planar}, the canonical orientation has the same weight $c$ on every vertex and we are able to obtain algorithms for the eight-vertex model under parameter settings where $c$ is relatively large, e.g. the region $\overline{\mathcal{C}}$; in \secref{sec:bipartite}, the canonical orientation has the same weight $d$ on every vertex and we are able to obtain algorithms for parameter settings where $d$ is relatively large, e.g. region $\overline{\mathcal{D}}$.
In general, the paradigm proposed in this paper can be applied to the study of the eight-vertex model on other classes of graphs in additional to planar/bipartite/torus graphs.
In particular, the methodology can be readily extended to any class of graphs with a ``canonical'' even orientation where every vertex has the same weight (one of $a$, $b$, $c$, or $d$).
}

\clearpage
\appendix
%\section*{Appendix}

%\section{}\label{app:complexity_table}
%
%\begin{table*}[htpb]
%\centering
%\caption[Caption for LOF]{Approximation complexity of the eight-vertex model on general (not necessarily planar) 4-regular graphs with $(a, b, c, d) \not\in \texttt{$d$-SUM}$\protect\footnotemark.}
%\ra{1.7}
%\begin{tabular}{@{}l@{\phantom{abc}}l@{\phantom{abc}}l@{}}
%\toprule
% & $d = 0$ & $d > 0$ \\
%\midrule
%$a=b=c=0$ & P-time computable (trivial) & P-time computable (trivial) \\
%$a=b=0, c>0$ & P-time computable (trivial) & \pbox{20cm}{$c=d$: P-time computable~\cite{DBLP:journals/corr/CaiF17} \\ $c\neq d$: NP-hard~\cite{DBLP:journals/corr/abs-1811-03126}} \\
%$a=0, b, c >0$ & NP-hard~\cite{DBLP:journals/corr/abs-1811-03126} & \#\textsc{PerfectMatchings}-hard (in this paper)\\
%$a, b, c > 0$ & NP-hard~\cite{doi:10.1137/1.9781611975482.136} & \#\textsc{PerfectMatchings}-hard (in this paper)\\
%\bottomrule
%\end{tabular}
%\end{table*}
%
%\footnotetext{There is a symmetry among $a, b, c$ for the eight-vertex model on general (not necessarily planar) 4-regular graphs, so for simplicity we assume $a \le b \le c$.}

\section{}\label{app:holo_trans}

The readers may have noticed that even though 
 $Z = \frac{1}{\sqrt{2}}\left[\begin{smallmatrix} 1 & 1 \\ i & -i \end{smallmatrix}\right]$ is a complex-valued matrix, under the  $Z$-transformation
not only the binary {\sc Equality} function $(=_2)$ is transformed
to the binary {\sc Disequality} function $(\neq_2)$,
the arity 4 constraint function $f$ is also transformed to a
\emph{real-valued} constraint function $Z^{\otimes 4} f$.

This is not a coincidence, but a consequence of the fact that
$f$ satisfies arrow reversal symmetry. 

% which in truth table
%is just the vector $(1, 0, 0, 1)$ indexed by $\{00,01,10,11\}$,
%is transformed to the  binary {\sc Disequality} function $(\neq_2)
%= (0, 1,1, 0)$, namely
%\[(1, 0, 0, 1) Z^{\otimes 2} = (0, 1,1, 0),\]
%which can be verified in matrix forms
%$Z^{T} I Z = \frac{1}{2} 
%\left[\begin{smallmatrix} 1 & i \\ 1 & -i \end{smallmatrix}\right]
%\left[\begin{smallmatrix} 1 & 0 \\ 0 & 1 \end{smallmatrix}\right]
%\left[\begin{smallmatrix} 1 & 1 \\ i & -i \end{smallmatrix}\right]
%= 
%%\left[\begin{smallmatrix}  0 & 1 \\ 1 & 0 \end{smallmatrix}\right]$,

We say a real-valued constraint function $f$ satisfies arrow reversal symmetry
if for all  $(a_1, \ldots, a_n) \in \{0, 1\}^n$,
\[f(a_1, \ldots, a_n) 
= f(\overline{a_1}, \overline{a_2}, \ldots, \overline{a_n}),\]
where $\overline{a_i} = 1 - a_i$ for all $i$.
\begin{lemma}
A real-valued $f$ of arity $n$ satisfies arrow reversal symmetry,
 if and only if $Z^{\otimes n} f$ is real-valued.
\end{lemma}
\begin{proof}
Suppose $f$ satisfies arrow reversal symmetry. Denote by $\widehat{f}
= Z^{\otimes n} f$.
We have
$2^{n/2}\widehat{f}=
 \left[\begin{smallmatrix} 1 & 1 \\ i & -i \end{smallmatrix}\right]^{\otimes n}
f$, and thus for  all $(a_1, \ldots, a_n) \in \{0, 1\}^n$,
\[2^{n/2}\widehat{f}_{a_1 \ldots a_n}
=\sum_{(b_1, \ldots, b_n) \in \{0, 1\}^n }
f_{b_1, \ldots, b_n}
\prod_{1 \le j \le n}
\left\{(-1)^{ a_j  b_j}i^{a_j}
\right\}.\]
Hence, taking complex conjugation,
\begin{eqnarray*}
2^{n/2}\overline{\widehat{f}_{a_1 \ldots a_n}}
&=&
\sum_{(b_1, \ldots, b_n) \in \{0, 1\}^n }
{f_{b_1 \ldots b_n}}
\prod_{1 \le j \le n}\left\{
(-1)^{ a_j  b_j}
(-i)^{ a_j } \right\}\\
&=&
\sum_{(c_1, \ldots, c_n) \in \{0, 1\}^n }
f_{c_1 \ldots c_n}
\prod_{1 \le j \le n}\left\{
(-1)^{a_j  (1-c_j)}
(-i)^{a_j}\right\}
\\
&=& 2^{n/2}\widehat{f}_{a_1 \ldots a_n}.
\end{eqnarray*}

Now in the opposite direction,
suppose 
$\widehat{f}$ is real.
We have $Z^{-1} = \frac{1}{\sqrt{2}}
\left[\begin{smallmatrix} 1 & -i \\ 1 & i \end{smallmatrix}\right]$, hence
by the inverse transformation
$2^{n/2}f=
 \left[\begin{smallmatrix} 1 & -i \\ 1 & i \end{smallmatrix}\right]^{\otimes n}
\widehat{f}$, and thus for  all $(a_1, \ldots, a_n) \in \{0, 1\}^n$,
\[2^{n/2}f_{a_1 \ldots a_n}
=\sum_{(b_1, \ldots, b_n) \in \{0, 1\}^n }
\widehat{f}_{b_1, \ldots, b_n}
\prod_{1 \le j \le n}
\left\{(-1)^{ a_j  b_j}(-i)^{b_j}
\right\}.\]
So
\begin{eqnarray*}
2^{n/2}f_{\overline{a_1} \ldots \overline{a_n}}
&=& 
\sum_{(b_1, \ldots, b_n) \in \{0, 1\}^n }
\widehat{f}_{b_1, \ldots, b_n}
\prod_{1 \le j \le n}
\left\{(-1)^{ (1-a_j)  b_j}(-i)^{b_j}
\right\}\\
&=& 
\sum_{(b_1, \ldots, b_n) \in \{0, 1\}^n }
\widehat{f}_{b_1 \ldots b_n}
\prod_{1 \le j \le n}\left\{
(-1)^{a_j  b_j}
i^{b_j}\right\}
\\
&=&
\overline{2^{n/2}f_{a_1 \ldots a_n}}\\
&=&
2^{n/2}f_{a_1 \ldots a_n}.
\end{eqnarray*}
\end{proof}

\clearpage

\bibliography{reference}{}
\bibliographystyle{alpha}
%\bibliographystyle{plainnat}

%\clearpage
%\appendix
%\input{appendix.tex} 

\end{document}